\pdfoutput=1 %
\documentclass[runningheads, envcountsame, a4paper]{llncs}

\usepackage[ruled,linesnumbered,vlined]{algorithm2e}
\SetKw{continue}{continue}
\usepackage{mathtools} %
\usepackage{stmaryrd} %
\usepackage{amssymb} %
\usepackage{subcaption} %

\usepackage{hyperref}
\hypersetup{
    pdftitle={A Myhill-Nerode Characterization and Active Learning for One-Clock Timed Automata}, %
    pdfauthor={Kyveli Doveri, Pierre Ganty, B. Srivathsan}, %
    pdfkeywords={Timed languages, Deterministic Timed automata, Canonical representation, Myhill-Nerode equivalence, Active Learning}, %
    pdfencoding=auto,
    colorlinks,
    linkcolor={blue!70!black},
    citecolor={blue!70!black},
    urlcolor={blue!70!black}
}

\usepackage{tikz}
\usetikzlibrary{positioning}
\usetikzlibrary{automata}
\tikzset{initial text={}}

\DeclarePairedDelimiter\minor{\lfloor}{\rfloor}
\DeclarePairedDelimiter\major{\lceil}{\rceil}
\DeclarePairedDelimiter\denote{\llbracket}{\rrbracket}

\newcommand{\Aa}{\mathcal{A}}
\newcommand{\Bb}{\mathcal{B}}
\newcommand{\Cc}{\mathcal{C}}
\newcommand{\Ll}{\mathcal{L}}
\newcommand{\Tt}{\mathcal{T}}
\newcommand{\HI}{\mathcal{HI}}
\newcommand{\hi}{\mathit{hi}}
\newcommand{\region}{\mathit{region}}
\newcommand{\row}{\mathit{row}}
\newcommand{\rl}{\ensuremath{\mathbf{R}^{L}}}

\newcommand{\lstar}{\ensuremath{\mathsf{L^*}}}
\spnewtheorem*{lemmastar}{Lemma}{\bfseries}{\itshape}
\spnewtheorem*{propositionstar}{Proposition}{\bfseries}{\itshape}
\spnewtheorem*{theoremstar}{Theorem}{\bfseries}{\itshape}

\spnewtheorem{assumption}[theorem]{Assumption}{\bfseries}{\itshape}

\begin{document}

\title{A Myhill-Nerode Characterization and Active Learning for One-Clock Timed Automata}
\titlerunning{\(1\)-DTA: Myhill-Nerode characterization and active learning}%
\author{
	Kyveli~Doveri\inst{1}\orcidID{0000-0001-9403-2860} \and
	Pierre~Ganty\inst{2}\orcidID{0000-0002-3625-6003} \and
	B.~Srivathsan\inst{3,4}\orcidID{0000-0002-1351-8824}
}
\institute{
	Unaffiliated\quad\email{kyv\_4@hotmail.com}  \and %
	IMDEA Software Institute, Madrid, Spain\quad\email{pierre.ganty@imdea.org} \and %
	Chennai Mathematical Institute, Chennai, India\quad\email{sri@cmi.ac.in} \and %
	CNRS IRL 2000, ReLaX, Chennai, India %
}
\maketitle

\begin{abstract}
	We present a Myhill-Nerode style characterization for languages recognized by one-clock deterministic timed automata (\(1\)-DTA).
	Although there is only one clock, distinct automata may reset it differently along the same word.
	This adds a significant challenge in the search for a canonical automaton.
	Our characterization is based on a new perspective of \(1\)-DTAs in terms of ``half-integral'' words that they accept, along with the reset information encoded by them.
	We apply our results to develop \(\lstar\) style algorithms that learn the canonical \(1\)-DTA.

\end{abstract}

\section{Introduction}

One of the most fundamental results in the theory of finite automata is the Myhill-Nerode theorem, which gives a characterization of regular languages in terms of the so-called Nerode equivalence.
The Nerode equivalence is key to defining a canonical representation for regular languages given by their corresponding minimal deterministic finite automata.

Given a regular language \(L\), the Nerode equivalence relation \(\sim_L\) partitions the set of all words such that two words \(u\) and \(v\) are equivalent (\(u \sim_L v\)) if and only if, for all continuations \(w\), \(uw \in L\) if and only if \(vw \in L\).
Moreover, the Nerode equivalence is the coarsest among all equivalences \(\sim\) that satisfy two properties: \(u \sim v\) implies \(ua \sim va\) for all letters \(a\), and \(u \sim v\) implies \(u \in L\) iff \(v \in L\).
From there, the Myhill-Nerode theorem states that \(L\) is regular if and only if \(\sim_L\) has finitely many equivalence classes.
The Nerode equivalence provides a deep insight into the structure of regular languages and enables the definition of a canonical representation for regular languages by associating to each regular language a deterministic finite automaton (DFA) as follows.
Each state uniquely corresponds to a Nerode equivalence class.
Furthermore, the defined automaton is minimal since the Nerode equivalence is the coarsest among all equivalences induced by DFAs for the language.
The Nerode equivalence also plays a key role in the active learning of regular languages like Angluin's \(\lstar\) algorithm~\cite{angluinLearningRegularSets1987}.

Our focus in this paper is on developing a Myhill-Nerode characterization, a canonical representation and, finally, an active learning algorithm for \(1\)-clock deterministic timed automata.
The model of timed automata~\cite{AlurDTimedAutomata94} provides a de facto automaton model for real-time systems.
Verification of timed automata has been extensively studied, with industry-strength tools~\cite{behrmannUPPAAL04}.
Results about canonical representations or active learning algorithms have been more elusive.
However, there has been a flurry of recent activity on active learning algorithms~\cite{anLearning1DTA2020,wagaActiveLearningDeterministic2023,vaandragerInfComp23,bruyereMMT24,DoveriIRTALearning2024} where the main challenge revolves around identifying the guards and resets of the transitions in the timed automaton.
Short of solving the problem, research works have considered restricted versions of timed automata.

The most common restriction is to look at automaton models with \emph{input-driven resets}: the clock values are determined by the timed word, and not by the automaton.
In Event-Recording Automata (ERA)~\cite{alurECA99}, there is a clock \(x_a\) attached to every letter \(a\) in the alphabet.
There are no other clocks.
Clock \(x_a\) records the time since the last occurrence of \(a\).
For instance, on reading a timed word \((2 \cdot a) (1.3 \cdot b)\) (where \(1.3\) is the time between \(a\) and \(b\)), we will have \(x_a = 1.3\) and \(x_b = 0\), no matter which ERA reads it.
An \(\lstar\)-like algorithm for learning ERAs was first proposed in 2010~\cite{grinchteinLearningEventrecordingAutomata2010}.

In the model of integer reset timed automata (IRTA)~\cite{sumanTimedAutomataInteger2008}, clocks can be arbitrary.
However, every transition that resets a clock, say \(x\), must have a guard of the form \(x=c\), where \(c\) is a natural number.
Therefore, every reset occurs at an integer time point.
Hence, the fractional values of all clocks are equal.
This makes it possible to reduce every IRTA to a language equivalent single-clock IRTA.
It still does not (yet) induce an input-driven reset.
However, the work of Bhave et al.~\cite{bhaveAddingDenseTimedStack2017} tells us that we can assume, without loss of generality, a transition with a guard of the form \(x=c\) necessarily resets the clock.
Such a model is referred to as strict \(1\)-IRTA.
Hence, if we know the maximum constant \(K\) used in the automaton, and if we assume the automaton is complete with guards of the form \(x = 0\), \(0 < x < 1\), \(x = 1\),\ldots,\(x = K\), \(K < x\), the resets becomes input-driven.
Using these properties, a Myhill-Nerode characterization and an \(\lstar\) algorithm for strict \(1\)-IRTAs was first formulated in 2024~\cite{DoveriIRTALearning2024}.
In real-time automata, there is a single clock which is reset in every transition.
Hence, on reading a timed word \((t_1 \cdot a_1) (t_2 \cdot a_2) \cdots (t_n \cdot a_n)\), the value of the clock is \(0\), and while executing the transition corresponding to \(a_i\), the clock value is \(t_i\) (time between \(a_{i-1}\) and \(a_i\)).
This property has been used to give a Myhill-Nerode characterization and an \(\lstar\) algorithm~\cite{anLearningRTA2021,anLearningNondeterministicRealTime2021}.

A second type of restriction that appears in the literature on timed automata learning is to consider one-clock deterministic timed automata (\(1\)-DTA).
Here, the clock resets are not determined by the word and instead depend on the automaton.
This makes the search for canonicity and Myhill-Nerode style characterizations particularly challenging.
To the best of our knowledge, no Nerode equivalences or Myhill-Nerode characterizations exist for \(1\)-DTAs.
On the other hand, active learning algorithms for \(1\)-DTAs have been studied without going via the Nerode equivalence route.
As there is no notion of a canonical \(1\)-DTA for a language, the learning algorithms may produce different automata in different runs.
Two \(\lstar\)-like algorithms for \(1\)-DTAs were presented in 2020~\cite{anLearning1DTA2020}.

More recently~\cite{wagaActiveLearningDeterministic2023}, in an exceptional result which considers no restriction on the model, a Myhill-Nerode characterization and an \(\lstar\)-like algorithm for deterministic timed automata (DTA) has been proposed.
The algorithm considers \emph{elementary languages}, which are sets of timed words, as the basis for the equivalence.
For example \(\{ (t_0 \cdot a)(t_1 \cdot b) \mid t_0 + t_1 = 1\}\) is an elementary language.
The paper proposes a Nerode equivalence over elementary languages and proves that for deterministic timed languages, this equivalence induces finitely many classes.
Using a monoid-based representation of deterministic timed languages~\cite{malerRecognizableTimedLanguages2004a}, a deterministic timed automaton is constructed.
However, there is no counterpart to the segment of the classical Myhill-Nerode theorem which says that the Nerode equivalence is the coarsest among all equivalences induced by DFAs.
Hence there is no concept of a canonical deterministic timed automaton for a languague.

In an orthogonal line of work, timers have been considered in the place of clocks~\cite{vaandragerInfComp23,bruyereMMT24}.
A timer can be set to an integer value in a transition, and its value decreases along with time.
The associated \emph{timeout} is observable by a series of experiments.
This facilitates a learning algorithm.
We are not aware of any Myhill-Nerode theorems or notions of canonicity for this model.

\paragraph{Our contributions.}
In the works presented above~\cite{wagaActiveLearningDeterministic2023,anLearning1DTA2020,vaandragerInfComp23,bruyereMMT24}, several models are on one-clock (or one-timer) automata.
Hence, this restriction to a single clock assumes substantial significance in practice.
In this work, we study \(1\)-DTAs and provide a fresh perspective which leads to a Myhill-Nerode characterization, a canonical \(1\)-DTA and a novel \(\lstar\) algorithm to learn the canonical \(1\)-DTA.
The first and foremost technical challenge lies in handling the clock resets.
We propose an automaton-independent way of viewing resets as \emph{reset functions} and identify a canonical reset function for each language.
A second key idea is to restrict attention to \emph{half-integral words} inside a language.
These are timed words where the delays are of the form \(\frac{n}{2}\) for some natural number \(n\).
We show that \(1\)-DTA languages are determined by the half-integral words that they accept and the canonical reset function.
Equipped with these new fundamental insights, we develop technical machinery for the Myhill-Nerode characterization and extend it to a learning algorithm.

\section{Background}
\label{sec:background}
\paragraph{Words and Languages.} %
An \emph{alphabet} is a finite set of \emph{letters} which we typically denote by \(\Sigma\).
An \emph{untimed word} is a finite sequence \(a_{1}\cdots a_{n}\) of letters \(a_{i}\in \Sigma\).
As usual, we denote the empty word by \(\epsilon\).
We denote by \(\Sigma^{*}\) the set of untimed words over \(\Sigma\).
An \emph{untimed language} is a subset of \(\Sigma^{*}\).
We denote non-negative reals by \(\mathbb{R}_{{}\geq 0}\) and the set of all finite sequences of these numbers by \(\mathbb{T}\).
A \emph{timed word} is a finite sequence \((t_1\cdot a_{1})\cdots (t_n\cdot a_{n})\) where \(a_{1}\cdots a_{n}\in \Sigma^{*}\) and \(t_1\cdots t_n\in \mathbb{T}\).
The value \(t_i\) denotes the delay between \(a_{i-1}\) and \(a_i\), for \(i > 1\).
We denote the set of timed words by \( \mathbb{T}\Sigma^{*}\).
A \emph{timed language} is a set of timed words.

The \emph{residual language} of a (un)timed language \(L\) with regard to a (un)timed word \(u\) is defined as \(u^{-1}L=\{w\mid uw \in L\}\).
In particular \(\epsilon^{-1} L = L\) for every (un)timed language \(L\).
Timed automata are recognizers of timed languages~\cite{AlurDTimedAutomata94}.
We focus on one-clock timed automata subclasses in this work.
Hence we present a modified definition suitable for our work.

\paragraph{One-Clock Timed Automata.} %
We assume the presence of a single clock \(x\).
A \emph{One-clock Timed Automaton} (\(1\)-TA) over \(\Sigma\) is a tuple \(\Aa = (Q, q_{I},T,F)\) where \(Q\) is a finite set of states, \(q_{I}\in Q\) is the initial state, \(F\subseteq Q\) is the set of final states and \(T\subseteq Q\times Q\times \Sigma \times \Phi \times \{0,1\}\) is a finite set of transitions where \(\Phi\) is the set of clock constraints given by \[ \phi ::= c < x \quad|\quad x = c \quad|\quad x < c \quad|\quad \phi \wedge \phi \enspace,\text{ where }c\in \mathbb{N} \text{ (natural numbers) }.
\]
For a clock constraint \(\phi\), we write \(\denote{\phi}\) for the set of values of \(x\) that satisfies the constraint.
In our syntax, a transition looks like \((q, q', a, \phi, r)\) where \(\phi\) is a clock constraint called the \emph{guard} of the transition and \(r \in \{0, 1\}\) denotes whether the single clock \(x\) is \emph{reset} in the transition: \(0\) denotes that it is reset, whereas \(1\) denotes otherwise.
The idea for this notation is that on reading the transition, the value of the clock is multiplied with \(r\).
When depicting \(1\)-TA (as in Section~\ref{sec:1dtas_HI_reset_function}), we use dashed edges for resetting transitions.

We say that a \(1\)-TA with transitions \(T\) is \emph{deterministic} whenever for every pair \(\theta=(q,q',a,\phi,r)\) and \( \theta_1=(q_1,q'_1,a_1,\phi_1,r_1)\) of transitions in \(T\) such that \(\theta\neq\theta_1\) we have that either \(q\neq q_1\), \(a\neq a_1\) or \(\denote{\phi}\cap\denote{\phi_1}=\emptyset\).
We write \(1\)-DTA for a one-clock deterministic timed automaton.

A \emph{run} \(e\) of \(\Aa\) on a timed word \((t_1\cdot a_{1})\ldots (t_k\cdot a_{k})\in \mathbb{T}\Sigma^{*}\) is a finite sequence \[ e=(q_{0}, \nu_{0}) \xrightarrow{t_1,\theta_{1}} (q_{1}, \nu_{1})\xrightarrow{t_2,\theta_{2}}\cdots \xrightarrow{t_k,\theta_{k}} (q_{k}, \nu_{k})\enspace, \] where \(q_{j}\in Q\), \(\nu_{j}\in \mathbb{R}_{\ge 0}\) for all \(j\in\{0,\ldots,k\}\) and, for each \(i \in \{1,\ldots,k\}\) the following hold: {\upshape(\itshape i\upshape)} \(\theta_i=(q_{i-1},q_{i},a_{i},\phi_{i},r_{i}) \in T\), {\upshape(\itshape ii\upshape)} \(\nu_{i-1}+t_i\in \denote{\phi_{i}}\), and {\upshape(\itshape iii\upshape)} \(\nu_{i} = (\nu_{i-1}+t_i)\times r_{i}\).
Therefore if \(r_i = 0\), we have \(\nu_i = 0\) and if \(r_i = 1\) we have \(\nu_i = \nu_{i-1} + t_i\).
A pair \( (q,\nu) \in Q\times\mathbb{R}_{\ge 0}\) like the ones occurring in the run \(e\) is called a \emph{configuration} of \(\Aa\) and the configuration \( (q_I,0) \) is called \emph{initial}.
The run \(e\) is deemed \emph{accepting} if \(q_{k}\in F\).
For \(w\in \mathbb{T}\Sigma^{*}\) we write \((q,\nu) \rightsquigarrow^w (q',\nu')\) if there is a run of \(\Aa\) on \(w\) from \((q,\nu)\) to \((q',\nu')\).
When \(\Aa\) is deterministic we note that every timed word has at most one run starting from the initial configuration.
We say that \(\Aa\) is \emph{complete} when it has exactly one run for every timed word.
Assuming \(\Aa\) is complete and deterministic, given a timed word \(u\) denote by \(\Aa(u)\) the unique state reached by the one run of \(\Aa\) on \(u\) from the initial configuration.

Finally, given a configuration \((q,\nu)\), define \(\Ll(q,\nu)=\{w\in \mathbb{T}\Sigma^{*}\mid (q,\nu) \rightsquigarrow^w (q',\nu'), q'\in F, \nu' \in \mathbb{R}_{\ge 0}\}\).
We write \(\Ll(\Aa)\) for \(\Ll(q_I,0)\).
We say that a language \(L\) is \(1\)-DTA recognizable if \(L = \mathcal{L}(\Aa)\) for some \(1\)-DTA \(\Aa\).

\paragraph{Equivalence Relation.}
A binary relation \(\mathord{\sim} \subseteq S\times S\) on a set \(S\) is an \emph{equivalence} if it is reflexive (i.e. \(y\sim y\)), transitive (i.e. \(y\sim z \wedge z\sim z' \implies y\sim z'\)) and symmetric (i.e. \(y\sim z\implies z\sim y\)).
The equivalence class of \(s\in S\) w.r.t. \(\sim\) is the subset \([s]_{\sim}=\{s'\in S\mid s\sim s'\}\).
A \emph{representative} of the class \([s]_{\sim}\) is any element \(s'\in [s]_{\sim}\).
Given a subset \(D\) of \(S\) we define \([D]_{\sim}=\{[d]_{\sim}\mid d\in D\}\).
We say that \(\sim \) has \emph{finite index} when \([S]_{\sim}\) is a finite set.

\paragraph{Region Equivalence for one clock.}
An important technical tool in the analysis of timed automata is the \emph{region equivalence}~\cite{AlurDTimedAutomata94}.
We recall this equivalence in the setting of \(1\)-TAs.
Define the equivalence \(\mathord{\equiv}\subseteq \mathord{\mathbb{R}_{\ge 0}\times \mathbb{R}_{\ge 0}}\) by \[y\equiv z \iff \minor{y} = \minor{z} \land (\{y\}=0 \Leftrightarrow \{z\}=0) \enspace ,\] where given \(y \in \mathbb{R}_{\ge 0}\), \(\minor{y} \) denotes its \emph{integral part} and \(\{y\}\) its \emph{fractional part}.
\footnote{Besides the floor function (\(\minor{y}\)), we use the ceiling notation (\(\major{y}\)) for convenience.}

Given a constant \(K\in\mathbb{N}\) define the equivalence \(\mathord{\equiv^{K}}\subseteq \mathord{\mathbb{R}_{\ge 0}\times \mathbb{R}_{\ge 0}}\) by \[y\equiv^{K} z \iff \bigl( y\equiv z \bigr) \lor (K<y \land K<z) \enspace .
\]
The equivalence \(\mathord{\equiv^{K}}\) is called the \emph{region equivalence w.r.t constant \(K\)} and its classes are called \emph{\(K\)-regions} (or regions when \(K\) is clear from the context).
\begin{assumption}\label{ass:guards-are-K-regions}
	Without loss of generality, we assume that guards of a \(1\)-TA are of the form \(x=c\) for \(c \in \{0, \dots, K\}\), \( c < x < c+1\) for \(c \in \{0, \dots, K-1\}\), or \(K<x\), where \(K\) is the largest constant appearing in the transitions of the automaton (in other words, every guard \(\phi\) is such that \(\denote{\phi}\) is a \(K\)-region).
\end{assumption}

Building on the above assumption, we introduce a notation that is useful for defining \(1\)-TAs.
Given \(t \in \mathbb{R}_{\ge 0}\) and \(K\in \mathbb{N}\) define the clock constraint \(\phi_K(t)\) as:
\[ \phi_K(t) =
	\begin{cases}
		x=t                       & \quad\text{if } t\leq K \wedge \{t\}=0\enspace,      \\
		\minor{t} < x < \major{t} & \quad\text{if } t\leq K \wedge  \{t\}\neq 0\enspace, \\
		K<x                       & \quad\text{else } (\text{i.e. }
		K < t)\enspace .
	\end{cases} \]

\section{Reset Functions}
\label{sec:reset-functions}

As said before, our goal in this work is to determine a machine-independent characterization for \(1\)-DTA languages.
A key challenge is to determine how the clock should be reset.
Different automata for the same language \(L\) may reset the clock differently along the same word.
In this section, we first abstract away from the underlying automaton and consider resets as a function from words to reals.
Then, we propose the \emph{syntactic reset function} which relies solely on the structure of \(L\), and is independent of any particular automaton.

A \emph{reset function} is a function \(R:\mathbb{T}\Sigma^{*}\to \mathbb{R}_{\geq 0}\) such that for every \(u \in \mathbb{T}\Sigma^{*}, t \in \mathbb{R}_{\geq 0}\) and \(a \in \Sigma\), we have \(R(\epsilon)=0\) and \(R(u (t \cdot a))\) is either \(0\) or \(R(u)+t\).
As its name suggests, a reset function prescribes how resets occur along a word: \(R(u (t \cdot a))=0\) corresponds to resetting after reading \(u (t \cdot a)\) whereas \(R(u (t \cdot a))=R(u)+t\) corresponds to not resetting.
Every (complete) \(1\)-DTA \(\Aa\) induces a reset function \(R^\Aa\) defined as \(R^\Aa(u) = y\) if \((q_{I},0) \rightsquigarrow^{u} (q,y)\).
\(R^\Aa(u)\) gives the value of the clock in \(\Aa\) after reading \(u\).

\subsection{Syntactic Reset Function}
\label{sec:syntactic-reset}
Given any arbitrary language \(L\), we want to define a machine-independent reset function \(\rl : \mathbb{T}\Sigma^{*} \to \mathbb{R}_{\geq 0}\) capturing the “resetting behaviour” of \(L\).
The main idea for defining \(\rl\) is that resets are determined by the residual languages of words.
Assume that \(L\) is \(1\)-DTA recognizable and fix a \(1\)-DTA \(\Aa\) accepting it.
If, after reading a word \(u\), \(\Aa\) resets the clock, then the residual language \(u^{-1}L\) is also recognized by a \(1\)-DTA: namely, by the same automaton \(\Aa\) with the initial state replaced by the state reached after reading \(u\).
Thus, resetting implies that \(u^{-1}L\) is \(1\)-DTA recognizable.
What about the converse?
That is, is there a \(1\)-DTA accepting \(L\) that resets after every word whose residual is \(1\)-DTA recognizable?
The answer is yes (our main result, Theorem~\ref{thm:computing-rl-acceptor}).
We define \(\rl\) so that for every word \(u\), \(\rl(u)=0\) iff \(u^{-1}L\) is \(1\)-DTA recognizable, whenever \(L\) is.

Formally, we define \(\rl\) inductively, prefix by prefix, starting from \(\epsilon\) as given next.
We set \(\rl(\epsilon)=0\), even when \(L\) is not \(1\)-DTA recognizable, since this definition is tailored to capture the behavior of \(1\)-DTA languages.
\begin{definition}[Syntactic reset function]
	Given a language \(L\subseteq   \mathbb{T}\Sigma^{*}\) we define the reset function \(\rl :  \mathbb{T}\Sigma^{*} \to \mathbb{R}_{\geq 0}\) inductively: \(\rl(\epsilon)=0\) and
	\[
		\rl(u(t\cdot a)) =
		\begin{cases}
			0 & \text{if \((u(t\cdot a))^{-1}
			L\) is a \(1\)-DTA recognizable language}, \\ \rl(u) +t & \text{otherwise}.
		\end{cases}
	\]
\end{definition}
\begin{example}
	Let \(L=\{(t_1 \cdot a)(t_2 \cdot b)\mid 0<t_1<1, t_1+t_2 =2\}\).
	Consider the word \((0 \cdot a)\).
	Notice that \(\rl(0\cdot a)\) must be zero since the word \((0\cdot a)\) has delay zero, and every reset function returns \(0\) on this word.
	For \(1\leq t_1\), \(\rl(t_1\cdot a)=0\) since \((t_1 \cdot a)^{-1}L=\emptyset\) is a \(1\)-DTA recognizable language and similarly \(\rl(u)=0\) if \(u\) is prefix of no word in \(L\).
	For \(0<t_1<1\), \((t_1\cdot a)^{-1}L=\{(t_2\cdot b) \mid ~ t_1+t_2 =2\}\) cannot be recognized using a clock initialized to zero, thus \(\rl(t_1\cdot a) =t_1\).
	For \(u\in L\), \(\rl(u) =0\) because \(u^{-1}L=\{\epsilon\}\) is \(1\)-DTA recognizable.
\end{example}

We now state our main theorem of this section which says that a \(1\)-DTA inducing the syntactic reset function \(\rl\) can in fact be computed starting from any \(1\)-DTA recognizing \(L\).
Along with it, we make another observation: the range of the syntactic reset function is bounded by a constant \(K_L\), and moreover, within this interval \([0, K_L]\), the range does not take any integral value other than \(0\).
Denoting by \(\mathbb{N}_{>0}\), the set of natural numbers greater than \(0\), we can summarize the range of \(\rl\) as \([0, K_L] \setminus \mathbb{N}_{>0}\).
Interestingly, this constant \(K_L\) turns out to be the smallest constant needed for any \(1\)-DTA to accept \(L\).
\begin{theorem}%
	\label{thm:computing-rl-acceptor}
	\hspace{0pt}
	\begin{itemize}
		\item Given a \(1\)-DTA for \(L\), we can compute an equivalent \(1\)-DTA with the same constant such that it induces the syntactic reset function \(\rl\).
		\item If \(L\) is a \(1\)-DTA recognizable language then for every timed word \(u\), \(\rl(u)\in [0,K_L]\setminus\mathbb{N}_{>0}\)  where \(K_L\) is the smallest constant \(K\) such that there exists a \(1\)-DTA with maximum constant \(K\) that accepts \(L\).
	\end{itemize}
\end{theorem}

\section{\(1\)-DTAs in Terms of Reset Functions and Half-Integral Words}
\label{sec:1dtas_HI_reset_function}

For our next step, we want to view \(1\)-DTA recognizable languages in terms of special discrete words.
In the untimed world, each run of a deterministic automaton corresponds to a unique word.
This correspondence does not hold for \(1\)-DTAs.
In the automaton \(\Aa_s\) below all timed words of the form \( (t_1 \cdot a) (t_2 \cdot b)\) such that \( 0 < t_1 < 1 \text{ and } 0 < t_1 + t_2 < 1 \) follow exactly the same run.

\hspace*{\stretch{1}}%
\begin{minipage}{.55\textwidth}
	\tikz[baseline=(1.south),node distance=2cm,auto,initial text=\(\Aa_s\), every edge/.append style={font=\scriptsize}]{
		\tikzstyle{every state}=[inner sep=1pt,minimum size=10pt]

		\node[initial,state] (1) {};
		\node[state] (2) [right=of 1] {};
		\node[accepting,state] (3) [right=of 2] {};

		\path[->] (1) edge node [above] {\(a, 0 < x < 1, 1\)} (2)
		(2) edge[dashed] node [above] {\(b, 0 < x < 1, 0\)} (3)
		;
	}
\end{minipage}\hfill\vspace{0pt}\\
To address this, some approaches use symbolic words~\cite{wagaActiveLearningDeterministic2023,grinchteinLearningEventrecordingAutomata2010}, where each letter represents a region and each word represents a set of timed words.
We introduce \emph{half-integral words} which are viewed as a canonical choice of a concrete timed word within a symbolic word.
The half-integral word accepted by \(\Aa_s\) is \((\frac{1}{2} \cdot a) (0 \cdot b)\).
\paragraph{Half-integral words.}
A timed word \((t_1 \cdot a_1) (t_2 \cdot a_2) \dots (t_n \cdot a_n)\) is \emph{half-integral} if each \(t_i\) (\(1 \le i \le n\)) is of the form \(t_i=\frac{m}{2}\) for some \(m\in\mathbb{N}\), that is, the fractional part \(\{t_i\}\) is either \(0\) or \(\frac{1}{2}\).
Also the empty word \(\epsilon\) is a half-integral word.
For a language \(L\), we write \(\HI(L)\) for the set of all half-integral words in \(L\).
For brevity, we write \(\HI\) instead of \(\HI(\mathbb{T}\Sigma^{*})\).

We observe an intriguing fact: a \(1\)-DTA language is almost entirely determined by the subset of its half-integral words.
We illustrate this with an example.
Consider the \(1\)-DTA \(\Aa_s\) above and the following \(1\)-DTA \(\Aa_q\):

\hspace*{\stretch{1}}%
\begin{minipage}{.55\textwidth}
	\tikz[baseline=(1.south),node distance=2cm,auto,initial text=\(\Aa_q\), every edge/.append style={font=\scriptsize}]{
		\tikzstyle{every state}=[inner sep=1pt,minimum size=10pt]

		\node[initial,state] (1) {};
		\node[state] (2) [right=of 1] {};
		\node[accepting,state] (3) [right=of 2] {};

		\path[->] (1) edge[dashed] node [above] {\(a, 0 < x < 1, 0\)} (2)
		(2) edge[dashed] node [above] {\(b, x=0,0\)} (3)
		;
	}
\end{minipage}\hfill\vspace{0pt}\\
The language \(\Ll(\Aa_s)\) equals \(\{ (t_1 \cdot a) (t_2 \cdot b) \mid 0 < t_1 < 1 \text{ and } 0 < t_1 + t_2 < 1 \}\) while \(\Ll(\Aa_q)\) equals \(\{ (t_1 \cdot a) (t_2 \cdot b) \mid 0 < t_1 < 1 \text{ and } t_2 = 0 \}\).
Clearly, the two languages do not coincide.
However, they coincide when restricted to half-integral words as \(\HI(\Ll(\Aa_s))=\HI(\Ll(\Aa_q))=\{(\frac{1}{2} \cdot a) (0 \cdot b)\}\).
The difference stems from the resets: \(\Aa_s\) does not reset its clock after \(a\), whereas \(\Aa_q\) resets it after \(a\).

This example illustrates our main observation: every \(1\)-DTA language \(L\) is uniquely determined by the set of its half-integral words \(\HI(L)\) \emph{and} the syntactic reset function \(\rl\).
Therefore, if we are given that \(\HI(L) = \{(\frac{1}{2} \cdot a) (0 \cdot b)\}\) and \(\rl\) does not reset after \(a\), then \(L = \Ll(\Aa_s)\), whereas for the same \(\HI(L)\), if \(\rl\) does reset after \(a\), then \(L = \Ll(\Aa_q)\).
To formalize this idea, we first define a map from timed words to half-integral words using a reset function.

\paragraph{From timed words to half-integral words, through a reset function.}
We now relate half-integral words and reset functions: given a reset function, we associate each timed word to a unique half-integral word such that every \(1\)-DTA conforming to the reset function has both the original word and the half-integral word thereof visit the same states along the exact same path.

Given a reset function \(R\), the \emph{normal form} \(N^{R}(u)\) of an arbitrary timed word \(u \in \mathbb{T}\Sigma^* \) is a half-integral word with the same untimed part as \(u\) and with shifted delays.
Given \(x\in \mathbb{R}_{\geq 0}\) let \(\hi(x)\) denote the unique half-integral value such that \(\hi(x)=x\) if \(\{x\}=0\) and, \(\hi(x)=\minor{x} +\frac{1}{2}\) if \(\{x\}\neq 0\).
For example, \(\hi(2.7) = 2.5\), \(\hi(1.2) = 1.5\), etc. Notice that \(\hi(x)\) is the unique half-integral value such that \(x \equiv \hi(x)\).
We extend the function \(\hi\) to sets as expected.
Building upon the function \(\hi\) we define the normal form by induction on the length of words: \(N^{R}(\epsilon)=\epsilon\) and \(N^{R}(u(t \cdot a))=N^{R}(u)(t' \cdot a)\) where \(t'=\frac{n}{2}\) for some \(n\in\mathbb{N}\) such that \(R(u) +t\equiv \hi(R(u)) +t'\).
For instance, if \(R((0.3 \cdot a)) = 0\), then \(N^R((0.3 \cdot a) (1.9 \cdot b)) = (0.5 \cdot a) (1.5 \cdot b)\); whereas if \(R((0.3 \cdot a)) = 0.3\), then \(N^R((0.3 \cdot a) (1.9 \cdot b)) = (0.5 \cdot a) (2 \cdot b)\).

\begin{remark}
	\label{remark:normalforHI}
	For \(w\in \HI\) we have \(N^{R}(w)=w\).
\end{remark}
\begin{proposition}
	\label{prop:normal_form}
	If \(R\) is the reset function of a \(1\)-DTA \(\Aa\), then for every word \(u\), the runs of \(u\) and \(N^{R}(u)\) in \(\Aa\) follow the same sequence of transitions.
	In particular, \(u \in L(\Aa)\) iff \(N^{R}(u) \in L(\Aa)\).
\end{proposition}
We now obtain the characterization previously announced: a \(1\)-DTA recognizable language is entirely determined by its half-integral words together with its syntactic reset function.
\begin{theorem}\label{thm:1-DTAs-as-hi-reset}
	If \(L_{1}\) and \(L_{2}\) are \(1\)-DTA recognizable languages then \(L_{1}=L_{2}\) iff \(\HI(L_{1})=\HI(L_{2})\) and \(\mathbf{R}^{L_{1}}= \mathbf{R}^{L_{2}}\).
\end{theorem}
\begin{proof}
	The left-to-right implication is immediate.
	Assume the right-hand side and let \(R=\mathbf{R}^{L_{1}}= \mathbf{R}^{L_{2}}\).
	By Theorem~\ref{thm:computing-rl-acceptor} there is a \(1\)-DTA accepting \(L_{1}\) whose reset function is \(R\), and the analogue holds for \(L_{2}\).
	By Proposition~\ref{prop:normal_form} for every word \(u\) we have \(u\in L_1 \iff N^{R}(u)\in L_{1}\) and \(u\in L_2 \iff N^{R}(u)\in L_{2}\).
	Since \(\HI(L_{1})=\HI(L_{2})\) we have \(N^{R}(u)\in L_1 \iff N^{R}(u)\in L_{2}\).
	Hence, \(u\in L_1 \iff u\in L_{2}\).
	\qed
\end{proof}
\section{Equivalences on Half-Integral Words}
\label{sec:from_equivalences_to_automata_and_back}
Here, we aim to further decompose our understanding of \(1\)-DTA recognizable languages.
As seen in the previous section, a \(1\)-DTA language is determined by a reset function together with a set of half-integral words.
We now focus on understanding the structure of this set of half-integral words, examining how a \(1\)-DTA organizes them.
We would like to describe this set based on an equivalence on half-integral words, in a similar way as the classical DFA case, where words are identified to be equivalent when they reach the same state.
Here, however, states are not enough to capture the future of words since this future also depends on the clock value the words reach.
Here is an example that illustrates this aspect.
Let \(L\) be the language recognized by the following \(1\)-DTA.

\hspace*{\stretch{1}}%
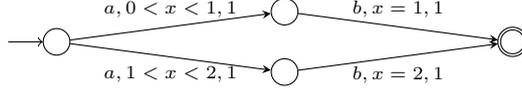
\begin{minipage}{.63\textwidth}
	\begin{tikzpicture}[state/.style={circle, draw, minimum size=10pt, inner sep=1pt, font=\scriptsize}, every edge/.append style={font=\scriptsize}] %
		\node[state, initial] (0) at (0,0) {};
		\node[state] (1) at (3, .4) {};
		\node[state] (2) at (3,-.4) {};
		\node[state,accepting] (3) at (6, 0) {};
		\begin{scope}[->, >=stealth, auto, anchor=center]
			\draw (0) edge node [above] {\(a, 0 < x < 1, 1\)} (1);
			\draw (0) edge node [below] {\(a, 1 < x < 2, 1\)} (2);
			\draw (1) edge node [above] {\(b, x = 1, 1\)} (3);
			\draw (2) edge node [below] {\(b, x = 2, 1\)} (3);
		\end{scope}
	\end{tikzpicture}
\end{minipage}\hfill\vspace{0pt}\\
Consider the half-integral words \(u = (0.5, a)\) and \(v = (1.5, a)\).
Their residuals are the same: \(u^{-1}L = v^{-1}L = \{(0.5, b)\}\).
However, if \(u\) and \(v\) go to the same state in a \(1\)-DTA, then the extension \((1.5, b)\) would also be possible from \(u\).
Thus, even though the residuals coincide, they need not go to the same state in a \(1\)-DTA.
Hence residuals do not identify states.
Here is the plan for this section.
\begin{enumerate}
	\item Following a similar reasoning on a subclass of \(1\)-DTAs called integer reset timed automata~\cite{DoveriIRTALearning2024}, we consider \(1\)-DTAs with a special syntax, which we call \(K\)-acceptors: every state is associated to a \(K\)-region and all words landing in the state end up with a clock value in this region.
	\item On the algebraic perspective, we consider a combination of equivalences over half-integral words and a reset function.
	      We call this pair a \emph{profile}.
\end{enumerate}
The main contribution of this section is in showing a correspondence between \(K\)-acceptors and profiles.
In the integer reset subclass~\cite{DoveriIRTALearning2024}, the clock resets were input-determined, that is, each word determined exactly when to reset.
For \(1\)-DTAs this is no longer the case, which is why we consider reset functions in addition to equivalences.
The presence of reset functions also adds new challenges in proving the correspondence between acceptors and profiles.

\paragraph{\(K\)-acceptors.}
\label{par:acceptors-and-strict-acceptors}
A \(K\)-\emph{acceptor} is a complete \(1\)-DTA with maximum constant \(K\) such that each state \(q\) is associated with a unique \(K\)-region, denoted \(\region(q)\).
Every timed word \(u\) that reaches \(q\) arrives with a clock value inside \(\region(q)\): \([R(u)]_{\equiv^{K}}=\region(q)\) where \(R\) is the reset function of the acceptor.
Syntactically, this means the following: if some incoming transition to \(q\) resets the clock, then every incoming transition resets it and \(\region(q) = \{0\}\); else, every incoming transition has a guard \(\phi\) such that \(\denote{\phi} = \region(q)\).
\begin{lemma}
	\label{lemma:acceptorbis}
	Given a \(1\)-DTA of constant \(K\) we can compute an equivalent \(K\)-acceptor with the same reset function.
\end{lemma}
\paragraph{Profiles.}
A \emph{profile} is a pair \((R, \approx)\) where \(R\) is a reset function, and \(\mathord{\approx}\subseteq \mathord{\HI\times \HI}\) is an equivalence relation on half-integral words.

Every acceptor \(\Aa\) induces a natural equivalence on half-integral words: \(u\approx^{\Aa}v\) if \(\Aa(u)=\Aa(v)\).
In other words, two half-integral words are equivalent whenever they lead to the same control state from the initial configuration.
The pair \((R^\Aa, \approx^\Aa)\) thus forms a profile.
But, what additional properties does \((R^\Aa, \approx^\Aa)\) satisfy?
This will help us understand the other direction: given an arbitrary profile \((R, \approx)\) when does it reflect an acceptor, in other words, when is there an acceptor \(\Aa\) such that \((R^\Aa, \approx^\Aa)\) equals \((R, \approx)\)?

The first step is to determine the maximum constant of an acceptor associated with \((R, \approx)\).
To do so, we associate to \((R, \approx)\) a constant \(K_{(R, \approx)} \in \mathbb{N} \cup \{\infty\}\) defined as follows.
Let \(C(K)\) be the following proposition \(\forall u \in \HI\), \(\forall a \in \Sigma\) and \(\forall n \in \mathbb{N}\), \(R(u) + \frac{n}{2} > K \implies u (\frac{n}{2} \cdot a) \approx u(K+\frac{1}{2}\cdot a)\).
Intuitively, \(C(K)\) says, the profile identifies \( u(\frac{n}{2}\cdot a)\) and \(u(K+\frac{1}{2}\cdot a)\) if after elapsing a delay \(\frac{n}{2}\) from \(R(u)\) the clock goes above \(K\).
Then \(K_{(R, \approx)}=\inf\{K\in \mathbb{N} \mid C(K) \text{ holds}\}\) if this infimum is defined and \(K_{(R, \approx)} = \infty\) otherwise.
As expected, for \((R, \approx)\) to define an acceptor, \(K_{(R, \approx)}\) must be finite.
Next we give the extra conditions required.
\begin{definition}[monotonic, \(L\)-preserving]
	\label{def:monotonic}
	A profile \((R, \approx)\) is \(L\)-\emph{preserving} when for all \(u, v \in \HI\), \(u\approx v \implies (u\in L \iff v\in L)\) and for all \(w \in \mathbb{T}\Sigma^*\), we have \(w \in L\) iff \(N^R(w) \in L\).
	It is said to have \emph{finite index} when \(\approx\) has finite index and \(K_{(R, \approx)}<\infty\).
	A profile \((R, \approx)\) is \emph{monotonic} if for every \(u, v \in \HI\), \(u \approx v\) implies {\upshape(\itshape a\upshape)} \(R(u) \equiv^{K_{(R, \approx)}} R(v)\) and {\upshape(\itshape b\upshape)} \(\forall n\in \mathbb{N}, \forall a\in \Sigma\), \(u (\frac{n}{2} \cdot a) \approx v (\frac{n}{2} \cdot a) \) and, the reset function \(R\) satisfies {\upshape(\itshape c\upshape)} \(R(w) \equiv R(N^R(w))\) for all \(w \in \mathbb{T}\Sigma^*\).
\end{definition}

\begin{proposition}\label{prop:automaton-to-monotone}
	For an acceptor \(\Aa\), the profile \((R^\Aa, \approx^\Aa)\) is monotonic, \(\mathcal{L}(\Aa)\)-preserving and has finite index.
\end{proposition}

For the converse, suppose \((R, \approx)\) is monotonic, \(L\)-preserving (for some \(L\)) and has finite index.
We construct an acceptor \(\Aa_{(R, \approx)}\) of constant \(K_{(R, \approx)}\) and prove that its language is \(L\).
The acceptor \(\Aa_{(R, \approx)}\) has states \(\{ [u]_\approx \mid u \in \HI \}\).
The initial state is \([\epsilon]_\approx\).
Final states are \(\{ [u]_\approx \mid u \in L\cap \HI \}\).
There is a transition \(([u]_\approx, [v]_\approx, a, g, s)\) in \(\Aa_{(R, \approx)}\) if there exists \(n \in\mathbb{N}\) such that:
\( u (\frac{n}{2} \cdot a) \approx v\), and \(g = \phi_K(R(u) + \frac{n}{2})\), and \(s = 0\) if \(R(v)=0\) and \(s = 1\) otherwise.
Next we establish the correspondence between acceptors and profiles.
\begin{proposition}\label{prop:monotone-to-automaton}
	Let \((R, \approx)\) be an \(L\)-preserving, monotonic profile with finite index.
	Then \(\Aa_{(R, \approx)}\) is a acceptor of profile \((R, \approx)\) such that \(\Ll(\Aa_{(R, \approx)}) = L\).
\end{proposition}
\section{Myhill-Nerode Style Characterization}
\label{sec:syntactic}
In this section, given a language \(L\), we define the \emph{syntactic profile} of \(L\), which is a machine-independent profile based on the syntactic reset function \(\rl\).
Using the syntactic profile and together with the results of the previous section, we obtain a Myhill-Nerode style characterization of \(1\)-DTA recognizable languages.
\paragraph{Syntactic Profile.}
The syntactic profile of \(L\) is \((\rl,\approx^{L})\) where \(\rl\) is the syntactic reset function of \(L\) and \(\approx^{L}\subseteq \HI\times \HI\) is the equivalence that identifies two half-integral words whenever their \(\rl\) values are equal and their residual languages are the same.
Formally, we define \[u\approx^{L}v \iff \rl(u) = \rl(v)~~\text{and}~~u^{-1}L = v^{-1}L \enspace.
\]

As shown next, when \(L\) is a \(1\)-DTA recognizable language \((\rl,\approx^{L})\) satisfies all the conditions of Definition~\ref{def:monotonic}.
We are thus in position to define the \emph{canonical acceptor} of \(L\) as \(\Aa_{(\rl, \approx^L)}\).
Moreover, the maximum constant of \(\Aa_{(\rl, \approx^L)}\) is \(K_{L}\) where we recall that \(K_L\) is the smallest constant \(K\) such that there exists a \(1\)-DTA with maximum constant \(K\) that accepts \(L\).
\begin{proposition}
	\label{prop:mh}
	If \(L\) is a \(1\)-DTA recognizable language then \((\rl,\approx^{L})\) is \(L\)-preserving, monotonic and has finite index such that \(K_{(\rl,\approx^{L})}=K_{L}\).
\end{proposition}
From Proposition~\ref{prop:monotone-to-automaton} and~\ref{prop:mh}, we get the first part of the Myhill-Nerode style characterization: \(L\) is \(1\)-DTA recognizable if and only if \((\rl,\approx^{L})\) is \(L\)-preserving, monotonic and has finite index.

The second part of a Myhill-Nerode style characterization concerns “minimality” of \(\Aa_{(\rl, \approx^L)}\).
Here, minimality refers to the number of states and the value of the maximum constant.
While the minimality of the constant is achieved by \(K_L\), minimality w.r.t.\ states is more involved.
To resolve this, we need to analyze the syntactic reset function \(\rl\) more closely which we proceed to do now.
\paragraph{Strictness.}
We single out a subclass of acceptors to which \(\Aa_{(\rl, \approx^L)}\) belongs called \emph{strict} acceptors.
Strict acceptors are a syntactic restriction of acceptors, analogous to the \emph{strict IRTAs} \cite{bhaveAddingDenseTimedStack2017} class: every transition with an equality guard, i.e.\ a guard of the form ``\(x = n\)'' for some \(n \in \mathbb{N}\), resets the clock.
Restricting to strict acceptors does not reduce expressiveness as shown next.

\begin{lemma}%
	\label{lemma:steptwo}
	Given a \(K\)-acceptor we can compute a language equivalent strict \(K\)-acceptor.
\end{lemma}
The next proposition shows that \(\Aa_{(\rl, \approx^L)}\) achieves minimality within the class of strict acceptors.
Example~\ref{example:counterexample_minimality} in Appendix~\ref{app:myhill-nerode} shows that \(\Aa_{(\rl, \approx^L)}\) is not necessarily minimal among all acceptors.

\begin{proposition}
	\label{prop:rl_is_strict}
	If \(L\) is a \(1\)-DTA recognizable language then \(\Aa_{(\rl, \approx^L)}\) is a strict \(K_{L}\)-acceptor with reset function \(\rl\).
	Moreover, no strict acceptor for \(L\) has fewer states than \(\Aa_{(\rl, \approx^L)}\).
\end{proposition}
Combining the results of Proposition~\ref{prop:monotone-to-automaton}, \ref{prop:mh} and \ref{prop:rl_is_strict} we finally get:
\begin{theorem}[Myhill-Nerode style characterization for \(1\)-DTA]
	\label{thm:good-equivalence-0}

	For a language \(L \subseteq \mathbb{T}\Sigma^{*}\) the following statements hold:
	\begin{itemize} %
		\item[{\upshape(\itshape a\upshape)}] \(L\) is \(1\)-DTA recognizable if and only if \((\rl,\approx^{L})\) is \(L\)-preserving, monotonic %
		      and has finite index.
		\item[{\upshape(\itshape b\upshape)}] No strict acceptor for \(L\) has fewer states or smaller constant than \(\Aa_{(\rl,\approx^{L})}\).
	\end{itemize}
\end{theorem}
\paragraph{Computing the canonical acceptor.}
\label{par:computing_the_canonical_acceptor}
Given a \(1\)-DTA \(\Aa\) for \(L\), we compute the canonical acceptor \(\Aa_{(\rl,\approx^L)}\) as follows.
By applying Theorem~\ref{thm:computing-rl-acceptor}, we transform \(\Aa\) into a \(1\)-DTA whose reset function is \(\rl\).
We then transform this \(1\)-DTA into an acceptor \(\Aa_{\rl}\).
This preserves the reset function so that the reset function of \(\Aa_{\rl}\) is also \(\rl\).
However, the equivalence on half-integral words induced by \(\Aa_{\rl}\) need not coincide with \(\approx^{L}\): there may be two half-integral words \(u\) and \(u'\) such that \(u \approx^{L}u'\) but \(u \not\approx^{\Aa_{\rl}}u'\).
Since \(u \approx^{L}u'\) we have that \(\rl(u)=\rl(u')\) and \(u^{-1}L=u'^{-1}L\).
Let \(\frac{n}{2}=\rl(u)\), \(q=\Aa_{\rl}(u)\) and \(q'=\Aa_{\rl}(u')\).
Then \(\region(q)=\region(q')\) and \(\mathcal{L}(q,\frac{n}{2})=\mathcal{L}(q',\frac{n}{2})\).
Conversely, if two states \(q\) and \(q'\) of \(\Aa_{\rl}\) have the same region and the same language when evaluated at the half-integral value of their region, this witnesses the existence of two half-integral words for which \(\approx^{L}\) and \( \approx^{\Aa_{\rl}}\) disagree.
Hence, to transform \(\Aa_{\rl}\) into an acceptor whose equivalence is \(\approx^{L}\) we apply on \(\Aa_{\rl}\) the following state merging algorithm.
For every pair of states \(q\) and \(q'\) such that \(\region(q)=\region(q')\) let \(\frac{n}{2}\) be the unique half-integral value in \(\region(q)\).
If \(\mathcal{L}(q,\frac{n}{2})=\mathcal{L}(q',\frac{n}{2})\), we merge \(q\) and \(q'\) by keeping one of them, say \(q\), redirecting all transitions entering \(q'\) to \(q\) and then deleting \(q'\) together with all its transitions.
Firstly notice that this modification preserves the set of half-integral words that are accepted.
Secondly, the redirection does not change the reset function, hence the modified automaton still implements \(\rl\).
Hence, from Theorem~\ref{thm:1-DTAs-as-hi-reset} the automaton obtained by merging the states still accepts the same language.
Deciding \(\mathcal{L}(q,\frac{n}{2})=\mathcal{L}(q',\frac{n}{2})\) reduces to deciding equivalence of \(1\)-DTAs as explained in Appendix~\ref{app:deciding_language_equality_for states}.

After merging all states that can be merged the resulting acceptor \(\Bb\) has profile \((\rl,\approx^L)\).
However, its constant may still be higher than \(K_L\).
The final step, is thus to adjust the constant.

Let \(K\) be the constant of \(\Bb\) and, \(T\) its set of transitions.
For a state \(q\) and a letter \(a\in \Sigma\) define \(q_a\) to be the target state of the transition with label \(a\) and guard ``\(K<x\)''.
Let \(C_{\Bb}(m)\) be the proposition ``for every state \(q\) and every letter \(a\), all \(a\)-labelled transitions from \(q\) requiring the clock to be strictly above \(m\) go to state \(q_a\)''.
Formally, \(C_{\Bb}(m): \forall (q,q',a,\phi_K(t),r)\in T\), \(t>m\implies q'=q_a\).
\begin{lemma}
	\label{lemma:computing_canonical_acc}
	\(C_{\Bb}(K_L)\) holds.
\end{lemma}
Due to Lemma~\ref{lemma:computing_canonical_acc} for every state \(q\) we group together all transitions of the form \((q,q_a,a,\phi_K(t),0)\) with \(K_L<t\) (all these transitions are necessarilly resetting since \(\region(q_a)=\{0\}\)) into one transition \((q,q_a,a,K_L<x,0)\).
The constant of the automaton thereof is \(K_L\).
We finish with examples related to Theorem~\ref{thm:good-equivalence-0}.
\begin{example}
	\hspace{0pt}
	\begin{enumerate}
		\item The language \(L=\{(t\cdot a) \mid t\in \mathbb{N}\}\) is not \(1\)-DTA recognizable.
		      For every \(K\) we have \((K +1 \cdot a) \in L\) but \((K+\frac{1}{2} \cdot a) \notin L\).
		      Thus, \(K_{(\mathbf{R}^{L}, \approx^{L})} =\infty \).
		\item The language \(L=\{(t_{1}\cdot a)(t_{2}\cdot b)\mid t_{1}\in (0,1)\backslash \{\frac{1}{2}\}, t_{1}+t_{2}<1\}\cup \{(\frac{1}{2}\cdot a)(0\cdot b)\}\) is not \(1\)-DTA recognizable as its syntactic profile is not monotonic due to condition (c) of Definition~\ref{def:monotonic}.
		      We have \(R^{L}(0.1\cdot a)=0.1\) and \(R^{L}(\frac{1}{2}\cdot a)=0\).
    \item Consider the language \(L=\{(t_{1}\cdot a)\dots (t_{n}\cdot a)\mid \exists i \le j, t_{i}+\dots + t_{j}=1\}\) of all timed words over the alphabet \(\Sigma=\{a\}\). 
          The residual \((0.3 \cdot a)^{-1}L\) is equal to the union of \(L\) and the language comprising all words whose delays add up to \(0.7\).
		      It cannot be recognized using a clock initialized to zero thus, \(\rl(0.3 \cdot a)=0.3\).
		      Hence, \(N^{\rl}((0.3 \cdot a)( 1.2 \cdot a))=( \frac{1}{2} \cdot a)( 1 \cdot a)\).
		      Since \(( \frac{1}{2} \cdot a)( 1 \cdot a)\in L\) and \((0.3 \cdot a)( 1.2 \cdot a)\notin L\) the syntactic profile of \(L\) is not \(L\)-preserving.
		      Thus, \(L\) is not \(1\)-DTA recognizable.
	\end{enumerate}
\end{example}

\section{Learning the Canonical Acceptor}
\label{sec:learning}
In this section, we apply our Myhill-Nerode characterization for active learning of \(1\)-DTA languages.
We present our Angluin-style \(\lstar\)~\cite{angluinLearningRegularSets1987} algorithms for learning the canonical acceptor \(\Aa_{(\rl,\approx^{L})}\) defined in Section~\ref{sec:syntactic}.
In this setting, a Learner aims to learn an unknown language \(L\) from a Teacher through \emph{membership} and \emph{equivalence} queries.
In a membership query, the Learner gives a word \(w\) to the Teacher, who responds saying whether \(w \in L\) or not.
In an equivalence query, the Learner gives a candidate automaton (a \emph{conjecture}) to the Teacher.
If the automaton accepts \(L\), the Teacher responds with a Yes, else the Teacher gives a counterexample to the conjecture.
All the information obtained using the queries is maintained by the Learner as an \emph{observation table}.

We first adapt observation tables to our setting, and then present two algorithms, in the same spirit as in the work of An et al.~\cite{anLearning1DTA2020} on \(1\)-DTA learning.
However, our objective is to learn the canonical acceptor and this is the new challenge which is not tackled by An et al.~\cite{anLearning1DTA2020}.
In our first algorithm, we assume a \emph{Smart Teacher} who can inform the Learner about the \(\rl\) function.
In this case, the task of the Learner lies in identifying \(\approx^L\).
The algorithm is almost similar to the classical \(\lstar\), except that now we need to also identify the unknown constant \(K\).
In the second algorithm, we remove the Smart Teacher assumption and use a technique, partly inspired by the approach of An et al.~\cite{anLearning1DTA2020}, that guesses all possible resets and converges to a strict acceptor for the language.
We show that using some additional equivalence queries, we can derive the canonical acceptor.

\paragraph{Observation tables.}
Given a finite subset \(S\subseteq \HI\) of half-integral words we define the ``one letter'' \emph{extensions} of \(S\) to be the words in \(S\Sigma_{K}\), where \(\Sigma_{K}=\hi([0,K+1)) \times\Sigma\).
Recall that \(\hi([0, K+1])\) stands for the set \(\{0, 0.5, \dots, K, K+0.5\}\).
An \emph{Observation Table (OT)} is a tuple \(\Tt=(K,S,E,T)\) where \(K\in \mathbb{N}\) is called \emph{the constant of \(\Tt\)}, \(S\subseteq \HI\) is a prefix-closed finite subset of half-integral words, \(E\subseteq \HI\) is a suffix-closed finite subset of half-integral words and \(T:(S\cup S\Sigma_{K}) \times E\to \{0,1\} \times \hi(\mathbb{R}_{\geq 0}\setminus\mathbb{N}_{>0})\).
\footnote{\(\mathbb{R}_{\geq 0}\setminus\mathbb{N}_{>0}\) is the set of real numbers minus the set of strictly positive integers.}

The \(\emph{reset assignment}\) of \(\Tt\) is the function \(r:(S\cup S\Sigma_{K})\times E \to \hi(\mathbb{R}_{\geq 0}\setminus\mathbb{N}_{>0})\) defined by \(r(s,e)=T(s,e)(2)\) where \(T(s,e)(i)\) denotes the \(i\)-th component of \(T(s,e)\).
We abuse notations and write \(r(s)\) for \(r(s,\epsilon)\).

Intuitively, if \(L\) is the unknown language, we will compute OTs such that \(T(s,e)(1)=1\) if{}f \(se\in L\) and \(T(s,e)(2)\) is a guess for \(\rl(se)\).

We view \(\Tt\) as a table of entries \((x,y)\in \{0,1\}\times \hi(\mathbb{R}_{\geq 0}\setminus\mathbb{N}_{>0})\) with rows \(S\cup S\Sigma_{K}\) and columns \(E\) labeled by half-integral words.
We write \(\row(s)\) for the row vector in \(\Tt\) indexed by \(s\) and, \(\row(s)(e)\) for the entry \(T(s,e)\).
For \(i\in \{1,2\}\), \(\row(s)(e)(i)=T(s,e)(i)\).
Table~\ref{tab:exampleOTs} shows examples of OTs.
In the OT \(\Tt'_0\), the entry \((0, \frac{1}{2})\) for the word \(u = (\frac{1}{2} \cdot a) \cdot \epsilon\) says that \(u \notin L\) and the value of the clock on reading \(u\) is \(\frac{1}{2}\).
In the OT \(\Tt''_0\), \(\row((\frac{1}{2}\cdot a)(1+\frac{1}{2}\cdot a))=((1, 0))\).
\begin{definition}\label{def:observation-table-closed-consistent}
	An OT \(\Tt\) of constant \(K\) is \emph{valid} when its reset assignment has range included in \([0,K]\setminus\mathbb{N}_{>0}\).
	It is \emph{closed} if for every \(w \in S\Sigma_{K}\) there is \(s\in S\) such that \(\row(w)=\row(s)\).
	We say that \(\Tt\) is \emph{consistent} if
	\begin{itemize}%
    \item[{\upshape(\itshape a\upshape)}] for every \(s_1, s_2 \in S\) if the equality \(\row(s_1) = \row(s_2)\) holds then we find that\linebreak \(\row(s_1(t \cdot a))= \row(s_2(t \cdot a))\) for every \((t \cdot a) \in\Sigma_{K}\).
		\item[{\upshape(\itshape b\upshape)}] for every \(s(t \cdot a) \in S\cup S\Sigma_{K}\) if the inequality \(K<r(s)+t\) holds then\linebreak \( \row(s(t \cdot a))= \row(s(K+\frac{1}{2} \cdot a))\).
	\end{itemize}
\end{definition}

\paragraph{Conjecture of an OT.}
A closed, consistent and valid OT \(\Tt\) of constant \(K\) induces a strict \(K\)-acceptor \(\Aa_{{\Tt}}\).
The automaton \(\Aa_{{\Tt}}\) has states \(\{ \row(s) \mid s \in S\}\).
The initial state is \(\row(\epsilon)\) and a state \(\row(s)\) is final if \(\row(s)(\epsilon)(1)=1\).
There is a transition \(( \row(s), \row(s'), a, g, d)\) if there exists \((t \cdot a) \in \Sigma_{K} \) such that: \(\row(s(t \cdot a)) = \row(s')\), \(g = \phi_K(\mathit{r}(s) + t)\) and \(d = 0\) if \(\mathit{r}(s')=0 \) and, \(d = 1\) otherwise.
Since the range of the reset assignments in an OT exclude \(\mathbb{N}_{>0}\), the resulting acceptor will in fact be strict, leading us to state this proposition.
\begin{proposition}
	\label{prop:acceptor}
	If \(\Tt\) is closed, consistent and valid, and has constant \(K\) then \(\Aa_{{\Tt}}\) is a strict \(K\)-acceptor such that \(R^{\Aa_{{\Tt}}}:\mathbb{R}_{\geq 0}\to [0,K]\setminus\mathbb{N}_{>0}\).
	Moreover, for every \(s\in S\), \((\row(\epsilon),0) \rightsquigarrow^{s} (\row(s),r(s))\)
\end{proposition}
\begin{table}[tbh]%
	\caption{Examples of Observation Tables}
	\label{tab:exampleOTs}
	\begin{minipage}{0.31\textwidth}
		\centering
		\subcaption{Initial with \(K=0\)}
		\begin{tabular}{|c|@{\hspace{2pt}}c@{\hspace{2pt}}|}
			\hline
			\(\Tt_{0}, K=0\)         & \(\epsilon\)        \\
			\hline
			\(\epsilon\)             & \((0,0)\)           \\
			\hline
			\((0\cdot a)\)           & \((0,0)\)           \\
			\((\frac{1}{2}\cdot a)\) & \((0,\frac{1}{2})\) \\
			\hline
		\end{tabular}
	\end{minipage}
	\hfill
	\begin{minipage}{0.21\textwidth}
		\centering
		\subcaption{With \(K=1\)}
		\begin{tabular}{|c|@{\hspace{2pt}}c@{\hspace{2pt}}|}
			\hline
			\(\Tt'_0, K=1\)            & \(\epsilon\)        \\
			\hline
			\(\epsilon\)               & \((0,0)\)           \\
			\hline
			\((0\cdot a)\)             & \((0,0)\)           \\
			\((\frac{1}{2}\cdot a)\)   & \((0,\frac{1}{2})\) \\
			\((1\cdot a)\)             & \((0,0)\)           \\
			\((1+\frac{1}{2}\cdot a)\) & \((0,0)\)           \\
			\hline
		\end{tabular}
	\end{minipage}
	\hfill
	\begin{minipage}{0.32\textwidth}
		\centering
		\subcaption{With \(K=1\)}
		\begin{tabular}{|c|@{\hspace{2pt}}c@{\hspace{2pt}}|}
			\hline
			\(\Tt''_0, K=1\)                                     & \(\epsilon\)        \\
			\hline
			\(\epsilon\)                                         & \((0,0)\)           \\
			\((\frac{1}{2}\cdot a)\)                             & \((0,\frac{1}{2})\) \\
			\hline
			\((0\cdot a)\)                                       & \((0,0)\)           \\
			\((1\cdot a)\)                                       & \((0,0)\)           \\
			\((1+\frac{1}{2}\cdot a)\)                           & \((0,0)\)           \\
			\((\frac{1}{2}\cdot a)(\{0,\frac{1}{2},1\}\cdot a)\) & \((0,0)\)           \\
			\((\frac{1}{2}\cdot a)(1+\frac{1}{2} \cdot a)\)      & \((1,0)\)           \\
			\hline
		\end{tabular}
	\end{minipage}%
\end{table}
\paragraph{Learning with a Smart Teacher.}
Our first learning algorithm, the \emph{Smart Teacher algorithm}, works under the assumption that the Teacher provides reset information.
More precisely, we assume that for every word \(u\in \mathbb{T}\Sigma^{*}\) the Teacher provides the value \(\rl(u)\).
Hence, in all the OTs computed the reset assignment corresponds to \(\rl\) that is for every row \(s\) and column \(e\), \(r(s,e)=\rl(s\cdot e)\).

The algorithm proceeds as follows.
The Learner starts with the OT \(\Tt_{0}=(0,\{\epsilon\},\{\epsilon\},T_0)\) like the one shown in Table~\ref{tab:exampleOTs}.
At each moment the Learner maintains an OT \(\Tt\) and executes the following main loop.
This loop uses the two auxiliary procedures \(\mathit{Process}(\Tt)\) and \(\mathit{Refine}(\Tt,w)\) where \(\Tt\) is an OT and \(w\in \mathbb{T}\Sigma^{*}\).
Both procedures return an OT and are described after the main loop.
\begin{itemize}
	\item if \(\Tt\) is not closed or not consistent or not valid then \(\Tt \gets \mathit{Process}(\Tt)\).
	\item if \(\Tt\) is closed, consistent and valid the Learner conjectures \(\Aa_{{\Tt}}\) and asks the Teacher whether the equivalence \(L(\Aa_{\Tt})=L\) holds.
	      \begin{itemize}
		      \item If \(L(\Aa_{{\Tt}})=L\) holds the algorithm terminates and returns \(\Aa_{{\Tt}}\).
		      \item Else, Teacher gives a counterexample \(w\in \mathbb{T}\Sigma^*\) and \(\Tt \gets \mathit{Refine}(\Tt,w)\).
	      \end{itemize}
\end{itemize}

\paragraph{Processing \(\Tt\).}
Given an OT \(\Tt=(K,S,E,T)\), \(\mathit{Process}(\Tt)\) outputs a closed, consistent valid OT of constant \(K'\geq K\).
It consists of applying the following \emph{four actions} until none of them is applicable.
\begin{enumerate} %
	\item If for \(w\in S\Sigma_{K}\) \(\row(w)=\row(s)\) for no \(s\in S\), move \(w\) from \(S\Sigma_{K}\) to \(S\).
	\item If for \(s_{1}, s_{2}\in S\), \(\row(s_1)=\row(s_2)\) but \(\row(s_1(t\cdot a))(e)\neq \row(s_2(t\cdot a))(e)\) for some \(e\in E\), add \((t\cdot a)e\) to \(E\).%
	\item If \(\Tt\) is not consistent w.r.t.\ item {\upshape (b)} in Definition~\ref{def:observation-table-closed-consistent}, increase \(K\) by one.
	\item If for \(s\in S\cup S\Sigma_{K}\) and \(e\in E\), \(r(s,e)>K\) , increase \(K\) by one.
\end{enumerate}
Action 1 corresponds to “closing” \(\Tt\), action 2 and 3 correspond to “making \(\Tt\) consistent” and action 4 to “making it valid”.
Notice that adding \((t \cdot a) e\) to \(E\) adds new entries that need to be filled using additional membership queries.
Similarly, incrementing \(K\) enlarges \(\Sigma_K\), and hence \(S \Sigma_K\).
Once again, new entries are created which need to filled using membership queries.

\paragraph{Refining \(\Tt\) with \(w\).}
Given an OT \(\Tt\) and a word \(w\in \mathbb{T}\Sigma^{*}\) (which may not be a half-integral word), \(\mathit{Refine}(\Tt,w)\) computes its normal form \(N^{\rl}(w)\) and then adds it along with its prefixes to \(S\).
Notice that if \(w\in \HI\) then \(N^{\rl}(w)=w\) (Remark~\ref{remark:normalforHI}) and that if \(w\notin \HI\), \(N^{\rl}(w)\) is computable from the values \(\rl(w')\) for every prefix \(w'\) of \(w\), which the Teacher provides.
\begin{proposition}
	\label{prop:ST}
	Smart Teacher algorithm terminates and returns \(\Aa_{(\rl, \approx^{L})}\).
	Moreover, the algorithm processes at most \(K_{L}+n\) OTs and the number of membership queries is \(O(mn^2 K_L |\Sigma|)\) where \(n\) is the number of equivalence classes for \(\approx^{L}\) and \(m\) the length of the longest counterexample.
\end{proposition}

\paragraph{Learning with a Normal Teacher.}
We now consider an algorithm, with no assumption on the Teacher.
The idea is similar to the Normal teacher algorithm of An et al.~\cite{anLearning1DTA2020}, but now, we are able to give minimality guarantees.
Since we do not have the reset information upfront, we need to guess all possible ways to reset the clock.
However, going down a wrong reset branch may not ensure termination.
Hence we need to explore a tree of OTs in a lock-step fashion, extending each leaf by one step.
The branch with the reset assignment corresponding to \(\rl\) produces a valid conjecture, and once we reach this node in the exploration tree, the algorithm terminates.

At each moment \(i\) the Learner keeps a pool \(P_{i}\) of OTs starting with the pool \(P_0=\{\Tt_{0},\dots, \Tt_{n}\}\) where all the \(\Tt_{i}\)'s are defined for \(K=0\), only contain the word \(\epsilon\) and differ from each other on the reset assignments such that all the possible reset assignments for the extensions of \(\epsilon\) appear in \(P_0\).
Given a pool \(P_{i}\) of OTs the Learner picks (and removes) one OT \(\Tt\in P_{i}\) and applies exactly \emph{one step} of the Smart Teacher algorithm: if \(\Tt\) is not closed or is inconsistent or is not valid, a single action among {\upshape 1--4} of the procedure \(\mathit{Process}(\Tt)\) is applied; if \(\Tt\) is closed, consistent and valid, the Learner makes a conjecture and processes the counterexample if one is returned, as explained next.

\paragraph{Processing counterexamples.}
The difference from the Smart Teacher algorithm is that if the Teacher returns a counterexample \(w\) that is not a half-integral word we do not know the values \(\rl(w')\) for the prefixes \(w'\) of \(w\) and thus cannot compute \(N^{\rl}(w)\).
Instead, the Learner computes all possible normal forms of \(w\) and adds them, along with all their prefixes, to the OT.
For example for \(w=(0.2 \cdot a)(1.3 \cdot a)\) there are two possible normal forms \(N_1(w)=(\frac{1}{2} \cdot a)(1+\frac{1}{2} \cdot a)\) and \(N_2(w)=(\frac{1}{2} \cdot a)(1 \cdot a)\).

Applying a single step of the Smart Teacher algorithm results in an extended OT, adding new entries — either new row words (from action~{\upshape 1} or from processing a counterexample), new suffix words (from action~{\upshape 2}), or new extensions from increasing the constant (action~{\upshape 3} and~{\upshape 4}).
For these new entries, the Learner must assign reset values.
To do so, it generates all possible reset assignments extending the assignment of \(\Tt\).
We require a possible reset assignment to satisfy \(r(\epsilon)=r(\epsilon,\epsilon)=0\), for every \(s(t\cdot a)\in S\cup S\Sigma_K\), \(r(s(t\cdot a))\in \{r(s)+t, 0\}\) and, for every \(u_1,u_2\in S\cup S\Sigma_K\) and every \(v_1,v_2\in E\) if \(u_1\,v_1 = u_2\, v_2\) then we require that \(r(u_1,v_1)=r(u_2,v_2)\).

Each such assignment yields a distinct OT; The Learner adds all these OTs to \(P_{i+1}\) and moves on to the next OT from \(P_{i}\).
Once all the OTs in \(P_{i}\) have been processed, \(P_{i}\) becomes empty and the new pool is \(P_{i+1}\).
Termination follows from the termination of the Smart Teacher algorithm and the fact that one of the branches corresponds to the execution of the Smart Teacher algorithm.
\begin{theorem}
	\label{theorem:normal_teacher}
	The Normal Teacher algorithm terminates and returns a strict acceptor for \(L\) of constant \(K\geq K_{L}\) and reset function \(R:\mathbb{R}_{\geq 0}\to [0,K]\setminus\mathbb{N}_{>0}\).
 	Moreover, the algorithm processes \(O(2^{H^{3}m2^m|\Sigma|}H)\) observation tables and performs \(O(2^{H^{3}m2^m|\Sigma|}H^3m2^m|\Sigma|)\) membership queries where \(H=K_{L}+n\), \(n\) is the number of equivalence classes for \(\approx^{L}\) and \(m\) the length of the longest counterexample.

\end{theorem}
\paragraph{Obtaining the canonical acceptor.}
Let \(\Tt\) be the final OT returned by the Normal Teacher algorithm and \((R, \approx)\) be the profile of \(\Aa_\Tt\).
A priori, the reset function \(R\) may differ from \(\rl\).
Consequently, \(\Aa_\Tt\) need not be isomorphic to \(\Aa_{(\rl,\approx^{L})}\).
We show next how to learn \(\Aa_{(\rl,\approx^{L})}\) from \(\Aa_\Tt\) using only equivalence queries.
We do it in two steps.
From $\Aa_\Tt$ we build a language equivalent acceptor $\Aa_{\rl}$ whose reset function is $\rl$, using Algorithm~\ref{alg:learning-rl-acceptor}.
Secondly, from $\Aa_{\rl}$ we perform a state merging to obtain the automaton which induces $\approx^L$ by the method explained in the subsection on computing canonical acceptor in Section~\ref{sec:syntactic}.
They key point is that deciding whether $q$ and $q'$ should be merged can be implemented by doing the merge and then checking whether the resulting automaton has an equivalent language, through an equivalence query to the Teacher.

We now explain Algorithm~\ref{alg:learning-rl-acceptor}.
Given a strict acceptor \(\Aa\) for \(L\) Algorithm~\ref{alg:learning-rl-acceptor} modifies one by one every state \(q\) at which the reset function of \(\Aa\) differs from \(\rl\) i.e., \(q=\Aa(u)\) for some \(u\in \mathbb{T}\Sigma^{*}\) and \(R^{\Aa}(u)\neq \rl(u)\).
Notice that necessarily \(\region(q)\neq \{0\}\).
The acceptor returned at the end has no such states, thus its reset function is \(\rl\).
The condition \(R^{\Aa}(u)\neq \rl(u)\) is checked by the equivalence query \(\Ll(\Aa_q) = L\) where \(\Aa_q\) is a modification of the automaton which only affects transitions and the region of state \(q\).
This region now becomes \(\{0\}\).
As a result, the number of states at which the reset function differs from \(\rl\) strictly decreases.

Given a strict \(K\)-acceptor \(\Aa\), and a state \(q\) such that \(\region(q) \neq \{0\}\), we define the strict acceptor \(\Aa_q\) obtained by modifying \(\Aa\) as follows:
\begin{itemize}
  \item Replace every transition \((q', q, a, \phi, 1)\) with \(q' \neq q\) in \(\Aa\) (there are no transitions incoming to \(q\) where clock is reset) by the transition \((q', q, a, \phi, 0)\) (clock is necessarily reset).
	\item Remove all transitions whose source is \(q\).
  \item Add transitions \(\{(q, q_a, a, x=0, 0),(q, q_a, a, 0<x<1, 0),\ldots,(q, q_a, a, x=K, 0),(q, q_a, a, K<x, 0)\}\), where \(q_a\) is the state corresponding to the transition \((q, q_a, a, x = K, 0)\) in \(\Aa\).
\end{itemize}
\begin{algorithm}[htb]
	\caption{Learning a acceptor with reset function \(\rl\)}
	\label{alg:learning-rl-acceptor}
	\KwIn{A strict acceptor \(\Aa\) for \(L\)}
	\KwOut{An acceptor with reset function \(\rl\)}
	\ForEach{state \(q\) in \(\Aa\) with \(\region(q)\neq\{0\}\)}{
		\lIf{Teacher confirms \(\Ll(\Aa_q) = L\)}{change \(\Aa\) to \(\Aa_q\)}
	}
	\Return \(\Aa\)\;
\end{algorithm}

Correctness of Algorithm~\ref{alg:learning-rl-acceptor} is given in Appendix~\ref{app:eliminating_bad_states}.

For the complexity: Algorithm~\ref{alg:learning-rl-acceptor} requires a number of equivalence queries equal at most to the number of states of \(\Aa_\Tt\) which is bounded by \(Hm2^m\) (by proof Theorem~\ref{theorem:normal_teacher}).
Minimization of \(\Aa_{\rl}\) uses at most \((Hm2^m)^2\) equivalence queries.

\section{Conclusion}
We have presented a Myhill-Nerode characterization for languages recognized by \(1\)-DTAs.
Theoretically, this is appealing since we are now able to identify a canonical automaton for a \(1\)-DTA language \(L\) and are also able to learn it.
Moreover, the canonical automaton has the minimal number of states among all strict acceptors for \(L\).
The work~\cite{wagaActiveLearningDeterministic2023} on the full class of deterministic timed automata presents one aspect of a Myhill-Nerode characterization: a language \(L\) is recognized by a DTA iff a certain equivalence has finite index.
It does not go on to identify canonical candidate automata.
One interesting future direction is to examine their results specifically for \(1\)-DTAs, equipped with our new insights.
From a practical perspective, the algorithms for \(1\)-DTA learning~\cite{anLearning1DTA2020} were enhanced~\cite{DBLP:conf/atva/XuAZ22} using satisfiability solvers to encode the guesses for the resets.
Since we work with strict acceptors, resets are fixed at transitions with guards \(x = c\).
This means that, in principle, we need fewer guesses.
A comprehensive optimization of our learning algorithm based on techniques from DFA learning, implementing and comparing to existing tools on \(1\)-DTA learning, is an interesting direction for future work.

\bibliographystyle{splncs04}%
\clearpage
\appendix

\section{Notational Conventions}\label{app:conventions}

For clarity and brevity we sometimes abuse notations related to clock constraints, especially in pictures.
Indeed we may write \(0 \leq x\) for the guard of a transition, e.g. \((q, q', a, 0 \leq x, r)\), even though \(\denote{0\leq x}\) is no \(K\)-region.
Such a transition actually stands for a set of transitions which cover exactly \(\denote{0\leq x}\) using \(K\)-regions.
For instance if \(K=2\) then \((q, q', a, 0 \leq x, r)\) stands for \((q, q', a, 0=x, r)\), \((q, q', a, 0<x<1 , r)\), \((q, q', a, 1=x, r)\), \((q, q', a, 1<x<2 , r)\), \((q, q', a, 2 = x, r)\), and \((q, q', a, 2<x , r)\).
A similar reasoning applies for guards like \( x\notin (0,1) \) or \( x\neq 2\).

\section{Proofs from Section~\ref{sec:1dtas_HI_reset_function}}

\begin{propositionstar}[\ref{prop:normal_form}]
	If \(R\) is the reset function of a \(1\)-DTA \(\Aa\), then for every word \(u\),
	the runs of \(u\) and \(N^{R}(u)\) in \(\Aa\) follow the same sequence of transitions. 
	In particular, \(u \in L(\Aa)\) iff \(N^{R}(u) \in L(\Aa)\).
\end{propositionstar}
\begin{proof}
	By induction on the length of words we show that for every word \(u\) the runs of \(u\) and \(N^{R}(u)\) follow the same transitions, end in the same state and that moreover, \(R(u) \equiv R(N^{R}(u))\).
	The base case \(u=\epsilon\) is true since \(N^{R}(\epsilon)=\epsilon\).
	For the inductive step, assume the statement holds for some word \(u\).
	Let \(u(t\cdot a)\in \mathbb{T}\Sigma^{*}\) and \(N^{R}(u(t \cdot a))=N^{R}(u)(t' \cdot a)\).
	Since \(R(u) \equiv R(N^{R}(u))\) we have \(\hi(R(u))=R(N^{R}(u))\).
	Then, by definition of \(t'\), \(R(u) +t\equiv R(N^{R}(u)) +t'\), which implies that after elapsing time \(t\) from \(R(u)\) and after elapsing time \(t'\) from \(R(N^{R}(u))\) the same guards are satisfied.
	Hence, \(u(t \cdot a)\) and \(R(N^{R}(u(t \cdot a)))\) reach the same state of the \(1\)-DTA and \(R(u(t \cdot a))\equiv R(N^{R}(u(t \cdot a)))\).
	\qed
\end{proof}
\section{Proofs from Section~\ref{par:acceptors-and-strict-acceptors}}
\begin{lemmastar}[\ref{lemma:acceptorbis}]
	Given a \(1\)-DTA of constant \(K\) we can compute an equivalent \(K\)-acceptor.
\end{lemmastar}
\begin{proof}
	For every state \(q\), we make \(2K + 2\) copies: \(q_0, q_{(0,1)}, q_1, \dots, q_K, q_{(K, \infty)}\) corresponding to the \(K\)-regions.
	For every original transition \((p, q, a, \phi, r)\) in the \(1\)-DTA, we have transition in the \(K\)-acceptor as follows:
	\begin{enumerate}
		\item if \(r = 0\), then add \((p_i, q_0, a, \phi, r)\) for all copies \(p_i\),
		\item else, from all copies \(p_i\) of \(p\), add an outgoing transition \((p_i, q_{\denote{\phi}}, a, \phi, r)\).
	\end{enumerate}
	\qed
\end{proof}
\begin{propositionstar}[\ref{prop:automaton-to-monotone}]
	For an acceptor \(\Aa\), the profile \((R^\Aa, \approx^\Aa)\) is monotonic, \(\mathcal{L}(\Aa)\)-preserving and has finite index.
\end{propositionstar}
\begin{proof}
	Finite-index of \(\approx^{\Aa}\) comes from finiteness of the set of states.
	For \(K\) the maximum constant of \(\Aa\) \(C(K)\) holds, thus \(K_{(R^\Aa, \approx^\Aa)}\leq K\) is finite.
	Therefore, \((R^\Aa, \approx^\Aa)\) has finite index.
	\(\mathcal{L}(\Aa)\)-preservation comes from Proposition~\ref{prop:normal_form} and the fact that two equivalent words reach the same state, thus either both are in the language or both are not in the language.

	We now prove monotonicity.
	Item {\upshape(\itshape a\upshape)} holds because each state of \(\Aa\) is associated to a unique region: If \(u \approx^\Aa v\) then \(\Aa\) reaches the same state, say \(q\), on reading \(u\) and \(v\).
	Therefore, \(R^\Aa(u) \in \region(q)\) and \( R^\Aa(v)\in \region(q)\).
	Thus, \(R^\Aa(u) \equiv^{K_{(R^\Aa, \approx^\Aa)}} R^\Aa(v)\).
	For item {\upshape(\itshape b\upshape)}, assume \(u \approx^\Aa v\) and let \(n\in \mathbb{N}\) and \( a\in \Sigma\).
	Since, \(\Aa(u)=\Aa(v)\) and \(R^\Aa(u) = R^\Aa(v)\) we have \(R^\Aa(u)+\frac{n}{2} = R^\Aa(v)+\frac{n}{2}\).
	Therefore, the last transitions on \(u (\frac{n}{2} \cdot a)\) and \(v (\frac{n}{2} \cdot a)\) are the same and so the two words end up in the same state.
	Thus, \(\Aa(u (\frac{n}{2} \cdot a))=\Aa(v (\frac{n}{2} \cdot a))\) and so \(u (\frac{n}{2} \cdot a) \approx^\Aa v (\frac{n}{2} \cdot a)\).
	Item {\upshape(\itshape c\upshape)} is also a consequence of Proposition~\ref{prop:normal_form} (see its proof).
	\qed
\end{proof}
\begin{propositionstar}[\ref{prop:monotone-to-automaton}]
	Let \((R, \approx)\) be an \(L\)-preserving, monotonic profile with finite index.
	Then \(\Aa_{(R, \approx)}\) is an acceptor of profile \((R, \approx)\) such that \(\Ll(\Aa_{(R, \approx)}) = L\).
\end{propositionstar}
\begin{proof}
	An easy induction on the length of half-integral words shows that for every \(u\in \HI\), \(([\epsilon]_{\approx},0) \rightsquigarrow^{u} ([u]_{\approx},R(u)) \).
	In particular this imply that \(\forall u\in \HI,R^{\Aa_{(R, \approx)}}(u)= R(u)\).
	Moreover, by \(L\)-preservation, we deduce that \(\HI(L)=\HI(L(\Aa_{(R, \approx)}))\).
	We will show that \(R = R^{\Aa_{(R, \approx)}}\).
	Then the following equivalences hold: using \(R = R^{\Aa_{(R, \approx)}}\) and Proposition~\ref{prop:normal_form}, \[ w\in L(\Aa_{(R, \approx)}) \iff N^{R}(w)\in L(\Aa_{(R, \approx)}) \enspace.
	\]
	Since \(\HI(L)=\HI(L(\Aa_{(R, \approx)}))\),
	\[N^{R}(w)\in L(\Aa_{(R, \approx)})\iff N^{R}(w)\in L  \enspace,\]
	and by \(L\)-preservation
	\[N^{R}(w)\in L \iff w \in L  \enspace.\]
	Therefore, for every \(w\in \mathbb{T}\Sigma^{*}\), \(w\in L(\Aa_{(R, \approx)})\) if{}f \(w\in L \).
	Thus, \(L(\Aa_{(R, \approx)})=L\).

	To conclude the proof we need to show that \(R =R^{\Aa_{(R, \approx)}}\).
	Let \(R'=R^{\Aa_{(R, \approx)}}\).
	We have already shown that for every \(w\in \HI\), \(R(w) = R'(w)\) holds.
	We now need to show it also holds for words that are not half-integral.
	We proceed by induction on the length of the words.
	Assume it holds for words of length at most \(n\).
	Let \(u(t\cdot a)\in \mathbb{T}\Sigma^{*}\) be a word of length \(n+1\).
	We will show that \( R(u(t\cdot a)) \equiv R'(u(t\cdot a))\).
	Since \(R(u) = R'(u)\) (induction hypothesis), \(R(u(t\cdot a))=(R(u)+t)s\) and \( R'(u(t\cdot a))=(R'(u) +t)s'\) for some \(s,s'\in \{0,1\}\) this will imply \( R(u(t\cdot a)) = R'(u(t\cdot a))\).
	By monotonicity {\upshape(\itshape c\upshape)} of the profile \((R,\approx)\) and the profile of \(\Aa_{(R, \approx)}\) we have \( R(u(t\cdot a))\equiv R(N^{ R}(u(t\cdot a)))\) and \( R'(u(t\cdot a))\equiv R'(N^{ R'}(u(t\cdot a)))\).
	Hence, by transitivity, it suffices to show that \( R(N^{ R}(u(t\cdot a)))\equiv R'(N^{ R'}(u(t\cdot a)))\) (we actually have equality since both values are half-integral).
	By induction hypothesis for every prefix \(v\) of \(u\) we have \(R(v) = R'(v)\).
	This imply \(N^{ R}(u)=N^{R'}(u)\) and \(R(u)+t=R'(u)+t\).
	Since \(N^{ R}(u(t\cdot a))\) and \(N^{ R'}(u(t\cdot a))\) are defined based on these two elements we deduce \(N^{ R}(u(t\cdot a)) =N^{ R'}(u(t\cdot a))\).
	Thus, \(R(N^{ R}(u(t\cdot a))) = R'(N^{ R'}(u(t\cdot a)))\).
	\qed
\end{proof}

\section{Proof of Theorem~\ref{thm:computing-rl-acceptor}}
\label{app:constructing-rl-acceptor}
\begin{theoremstar}[\ref{thm:computing-rl-acceptor}]
	\begin{itemize}
		\item Given a \(1\)-DTA for \(L\), we can compute an equivalent \(1\)-DTA with the same constant such that it induces the syntactic  reset function \(\rl\).
		\item If \(L\) is a \(1\)-DTA recognizable language then for every timed word \(u\), \(\rl(u)\in [0,K_L]\setminus\mathbb{N}_{>0}\)  where \(K_L\) is the smallest constant \(K\) such that there exists a \(1\)-DTA with maximum constant \(K\) that accepts \(L\).
	\end{itemize}
\end{theoremstar}
Our proof of Theorem~\ref{thm:computing-rl-acceptor} proceeds as follows: starting from a \(1\)-DTA of constant \(K\) accepting \(L\) we first compute using Lemma~\ref{lemma:acceptorbis} an equivalent \(K\)-acceptor.
Let \(\Aa\) be this acceptor.
Our objective is to compute an acceptor for \(L\) whose reset function is \(\rl\) starting from \(\Aa\).
Our procedure first identifies the ``bad'' states of \(\Aa\) defined next.
\begin{definition}
	\label{def:bad_state}
	For a \(K\)-acceptor \(\Aa\) accepting \(L\) a state \(q\) is bad if \(q=\Aa(u)\) for some \(u\in \mathbb{T}\Sigma^{*}\) such that \(R^{\Aa}(u)\neq 0\) and \(u^{-1}L\) is a \(1\)-DTA recognizable language.
\end{definition}
Subsequently, each bad state is eliminated by a modification of the automaton that preserves the language.
In the end the resulting automaton has no bad states.
Therefore, by the next lemma, its reset function is \(\rl\).

\begin{lemma}\label{lemma:no-bad-rl}
	If \(\Aa\) is an acceptor for \(L\) with no bad states then \(R^{\Aa}=\rl\).
\end{lemma}
\begin{proof}
	By induction on the length of the words we show that \(\forall u\in \mathbb{T}\Sigma^{*}, R^{\Aa}(u)=\rl(u)\).
	It holds for \(\epsilon\).
	Let \(u(t\cdot a)\in \mathbb{T}\Sigma^{*}\).
	If \(R^{\Aa}(u(t\cdot a))=0\) then \(u(t\cdot a)^{-1}L\) is a \(1\)-DTA recognizable language, thus \(\rl(u(t\cdot a))=0\).
	If \(R^{\Aa}(u(t\cdot a))\neq 0\) then \(u(t\cdot a)^{-1}L\) is not a \(1\)-DTA recognizable language or otherwise \(\Aa(u(t\cdot a))\) would be a bad state.
	Using the induction hypothesis we find \(R^{\Aa}(u(t\cdot a))=R^{\Aa}(u)+t=\rl(u)+t\).
	Using the definition of \(\rl\) we find \(\rl(u(t\cdot a))=\rl(u)+t\).
	Thus, \(R^{\Aa}(u(t\cdot a))=\rl(u(t\cdot a))\).
\end{proof}

Using the two lemmas below we compute an equivalent strict \(K\)-acceptor whose reset function has range contained in \([0,K]\).

\begin{lemma}%
	\label{lemma:stepone}
	Given a \(K\)-acceptor we can compute an equivalent \(K\)-acceptor such that its reset function has range included in \([0,K]\).
\end{lemma}
\begin{proof}
	Let \(\Aa=(Q, q_{I},T,F)\) be a \(K\)-acceptor.
	We will compute a \(K\)-acceptor \(\Bb\) such that \(L(\Aa)= L(\Bb)\) and \(\Bb\) has no states \(q\) such that \(\region(q)=[K+\frac{1}{2}]_{\equiv^{K}}\).
	Then the reset function of \(\Bb\) has range included in \([0,K]\).
	In the following, we say that a state \(q\) is large if \(\region(q)=[K+\frac{1}{2}]_{\equiv^{K}}\).
	A state is small if it is not large in the previous sense.
	We modify \(\Aa\) as follows.
	\begin{enumerate}
		\item Replace every large state \(q\) by a new state \(\overline{q}\).
		\item Replace every transition \((p,q,a,K<x,1)\in T\) where \(p\) is small and \(q\) is large, by a transition \((p,\overline{q},a,K<x,0)\).
		\item Replace every transition \((s, s', a, K<x, 1)\in T\) where \(s\) and \(s'\) are large by a transition \((\overline{s}, \overline{s}', a, 0 \leq x, 0)\).
		\item Replace every resetting transition \((q, p, a, K<x, 0)\in T\) where \(q\) is large by a transition \((\overline{q}, p, a, 0\leq x, 0)\).
	\end{enumerate}
  Observe that the conventions of Section~\ref{app:conventions} apply for the last two items.
	We define the initial state of \(\Bb\) to be \(q_{I}\) and its accepting states to be all the states \(p\) and \(\overline{p}\) such that \(p\in F\).

	This modification removes all transitions of the form \((p,q,a,K<x,1)\) and only adds resetting transitions with target states of region \(\{0\}\).
	Thus, \(\Bb\) has no states \(q\) such that \(\region(q)=[K+\frac{1}{2}]_{\equiv^{K}}\).
	We now explain why \(L(\Aa)= L(\Bb)\).
	In a run of \(\Aa\), after reaching a large state and as long as we have not reset, we can only visit large states.
	This happens because when above the constant only transitions with guard ``\(K<x\)'' are enabled.
	Moreover, such transitions allow all time delays.
	Thus, such a ``large'' part of a run of \(\Aa\) is simulated by \(\Bb\) using the transitions of items {\upshape 3}.
	A transition from item {\upshape 2} corresponds to the first occurrence of a large state after a ``small'' portion of the run and, a transition from item {\upshape 4} corresponds to the last occurrence of a large state after a ``large'' portion of the run.

	To conclude, observe that the number of states of \(\Aa\) and \(\Bb\) coincide and so does their constants \(K\).
	\qed
\end{proof}
\begin{lemmastar}[\ref{lemma:steptwo}]%
	Given a \(K\)-acceptor we can compute a language equivalent strict \(K\)-acceptor.
\end{lemmastar}
\begin{proof}
	Let \(\Aa=(Q, q_{I},T,F)\) be a \(K\)-acceptor.
  We compute a \(K\)-acceptor \(\Bb\) such that \(L(\Aa)= L(\Bb)\) and \(\Bb\) has no states \(q\) with \(\region(q)=[m]_{\equiv^{K}}\) where \(m=\max(\{n\in\mathbb{N}\mid \exists p\in Q, \region(p)=[n]_{\equiv^K}\})\).
  We then repeat this construction: from \(\Bb\) we compute a \(K\)-acceptor with no states \(q\) with \(\region(q)=[m-1]_{\equiv^{K}}\) and so on until we obtain an acceptor with no states \(q\) with \(\region(q)\in\mathbb{N}_{>0}\).

	In the following we say that a state \(q\) is critical if \(\region(q)=[m]_{\equiv^{K}}\).
	We say it is good if it is not critical.
	Following Lemma~\ref{lemma:stepone} we can assume \(\Aa\) has no states \(q\) with \(\region(q)=[K+\frac{1}{2}]_{\equiv^{K}}\) (not a necessary assumption but limits the copies we add in item {\upshape 2} below).
	We modify \(\Aa\) as follows.

	\begin{enumerate}
		\item Replace every critical state \(q\) by a new state \(\overline{q}\) and remove all the transitions associated to \(q\).
		\item For every state \(p\) s.t. \(\region(p)=(i, i+1)\) for \( m\leq i\) add a copy \(\overline{p}\).
		\item For every transition \((p,q,a,x=m,1)\in T\) where \(p\) is good and \(q\) is critical add a transition \((p,\overline{q},a,x=m,0)\).
		\item For every transition \((s, s', a,x=m,1)\in T\) where \(s\) and \(s'\) are critical add a transition \((\overline{s}, \overline{s}', a, x=0, 0)\).
		\item For every transition \((s, s', a, \phi, 1)\in T\) s.t. \(\region(s)\in \{[m]_{\equiv^{K}}\}\cup \{[(i, i+1)]_{\equiv^{K}}\mid  m\leq i\}\) and \( \denote{\phi}=(c, c+1)\) for some \( m\leq c\) add a transition \((\overline{s}, \overline{s}', a, c -m< x  < c-m +1, 1)\).
		\item  For every transition \((s, s', a, \phi, 0)\in T\) s.t. \(\region(s)\in \{[m]_{\equiv^{K}}\}\cup \{[(i, i+1)]_{\equiv^{K}}\mid  m\leq i\}\) add a transition \((\overline{s}, s', a, \mathit{modify}(\phi), 0)\) where,
		      \[
			      \mathit{modify}(\phi)=
			      \begin{cases}
				      c-m < x < c-m +1 & \quad\text{if \(\phi\) is \(c < x  < c +1\)}      \\
				      K-m < x          & \quad\text{if \(\phi\) is \(K < x \)}             \\
				      x = c-m          & \quad\text{else (\(\phi\) is \(x=c\))} \enspace .
			      \end{cases}
		      \]
	\end{enumerate}
	Finally, we define the initial state of \(\Bb\) to be \(q_{I}\) and its accepting states to be all the states \(p\) and \(\overline{p}\) such that \(p\in F\).
  It is worth pointing that, following the conventions of Section~\ref{app:conventions}, modifying the guard \(K < x\) of a transition to \(K-m < x\) for \(m>0\) will result in adding several transitions each one associated with a \(K\)-region.

  We introduced no critical states \(q\) with  \(\region(q)\in\mathbb{N}_{>0}\) since the only non resetting transitions that we added (item {\upshape 5}) came with non equality guards.
	Since we removed all critical states (item {\upshape 1}) \(\Bb\), has no critical states.

	In a run of \(\Aa\) the first occurrence of a critical state is simulated by a transition in item {\upshape 1}.
	A ``critical portion'' of a run (meaning after reaching a critical state and as long as we haven't yet reset) is simulated thank to transitions in items {\upshape 4} and {\upshape 5}.
	Notice that in item {\upshape 5} guards are necessarily of the form \( \denote{\phi}=(c, c+1)\) for \(m \leq c\) because \(m\) is maximal and because we have assumed \(\Aa\) has no states \(q\) with \(\region(q)=[K+\frac{1}{2}]_{\equiv^{K}}\).
	Finally transitions in item {\upshape 6} are used to exit from a critical portion of the run.
	\qed
\end{proof}
\begin{remark}
	\label{remark:bad_states}
	The proofs of Lemma~\ref{lemma:stepone} and Lemma~\ref{lemma:steptwo} show that states with regions \( [1]_{\equiv^{K}}, [2]_{\equiv^{K}},\ldots, [K]_{\equiv^{K}}, [K+\frac{1}{2}]_{\equiv^{K}}\) are bad.
\end{remark}
In the acceptor we have computed so far every bad state has region of the form \((i,i+1)\) for some \(i\in \{0, \dots, K-1\}\).
Our next step is to remove these remaining bad states.
In Section~\ref{app:bad_states_and_their_properties} we present certain properties of bad states that we use to obtain in Section~\ref{app:characterization} a characterization of the remaining bad states.
In Section~\ref{app:eliminating_bad_states} we give an algorithm to eliminate the remaining bad states which is based on the characterization of Section~\ref{app:characterization}.
In Section~\ref{app:proof} we prove the second item of the theorem.

\subsection{Bad States and their Properties}
\label{app:bad_states_and_their_properties}
\begin{lemma}\label{lem:bad-state-same-language-for-all-x}
	Let \(q\) be a bad state in an acceptor.
	Then, \(\Ll(q, y) = \Ll(q, y')\) for all \(y, y' \in \region(q)\).
\end{lemma}
\begin{proof}
	Let \(\Bb\) be an acceptor and \(q\) a bad state in \(\Bb\).
	We have \(\region(q) \neq \{0\}\).
	Moreover, there is \(x \in \region(q)\) such that \(\Ll(q, x)\) is \(1\)-DTA recognizable.
	Let \(\Cc_{(q,x)}\) be a \(1\)-DTA that recognizes \(\Ll(q,x)\), and let \(p\) be the initial state of \(\Cc_{(q,x)}\).
	We have \(\Ll(q, x) = \Ll(p,0)\).
	We now modify \(\Bb\) and \(\Cc_{(q,x)}\) by one transition in each:
	\begin{itemize}
		\item add \((\bar{q}, q, a, \phi_{K}(\region(q)), 1)\) to \(\Bb\),
		\item add \((\bar{p}, p, a, \phi_{K}(\region(q)), 0)\) to \(\Cc_{(q,x)}\).
	\end{itemize}

	We will prove that \(\Ll(\bar{q}, 0) = \Ll(\bar{p}, 0)\) (this proof appears below).
	Before that, we show how the lemma becomes true if \(\Ll(\bar{q}, 0) = \Ll(\bar{p}, 0)\) is true.
	For any \(y \in \region(q)\), we have:
	\begin{align*}
		(y \cdot a)^{-1} \Ll(\bar{q}, 0) & = \Ll(q, y) & \text{and} &  & (y \cdot a)^{-1} \Ll(\bar{p}, 0) & = \Ll(p, 0)\enspace .
	\end{align*}
	Assuming \(\Ll(\bar{q}, 0) = \Ll(\bar{p}, 0)\), we get \(\Ll(q, y) = \Ll(p, 0)\).
	Coupled with \(\Ll(q, x) = \Ll(p, 0)\), we get \(\Ll(q, y) = \Ll(q,x)\) for all \(y \in \region(q)\).
	This proves the lemma.

	To prove \(\Ll(\bar{q}, 0) = \Ll(\bar{p}, 0)\) we will need the following claim.
	\begin{claim}
		Let \(x, x' \in \mathbb{R}_{\ge 0}\) such that \(x \equiv x'\).
		Then for every sequence \(t_1, t_2 \dots, t_m\), with \(m \ge 1\) and each \(t_i \in \mathbb{R}_{\ge 0}\), there exists a sequence \(t'_1, t'_2, \dots, t'_m\) such that: 
		\begin{itemize}
			\item \(x + t_1 + \cdots + t_j \equiv x' + t'_1 + \cdots + t'_j\) for all \(1 \le j \le m\), and
			\item \(t_i + t_{i+1} + \cdots + t_j \equiv t'_i + t'_{i+1} + \cdots + t'_j\), for all \(1 \le i \le j \le m\).
		\end{itemize}
	\end{claim}
	\begin{proof}
		Choose \(\lfloor t'_i \rfloor = \lfloor t_i \rfloor\) for every \(i\).
		The fractional values need to be chosen appropriately.
		Start with \(t'_1\).
		Consider the order between \(\{t_1\}\) and \(1 -\{x\}\): if \(\{t_1\} = 1 - \{x\}\), then set \(\{t'_1\} = 1 - \{x'\}\); if \(\{t_1\} < 1 - \{x\}\), choose a value for \(\{t'_1\}\) between \(0\) and \(1 - \{x'\}\); else choose a value in the interval \((1 - \{x'\}, 1)\).
		For \(t_2\), we consider two numbers \(1 - \{x + t_1\}\) and \(1 - \{t_1\}\) and compare the relation of \(\{t_2\}\) with these numbers.
		We choose \(\{t'_2\}\) in such a way so that the same relation is satisfied by it w.r.t \(1 - \{x' + t'_1\}\) and \(1- \{t'_1\}\).
		Indeed, assume we have picked \(t'_1, \dots, t'_{j-1}\), choose \(\{t'_j\}\) such that:
		\begin{itemize}
			\item \(1 - \{x + t_1 + \cdots t_{j-1}\} \le \{t_j\}\) if{}f  \(1 - \{x' + t'_1 + \cdots t'_j \} \le \{t'_j\}\),
			\item \(\{t_j\} \le 1 - \{x + t_1 + \cdots t_{j-1}\}\) if{}f \(\{t'_j\} \le 1 - \{x' + t'_1 + \cdots t'_j \}\),
			\item \(1 - \{t_i + t_{i+1} + \cdots t_{j-1}\} \le \{t_j\}\) if{}f  \(1 - \{t'_i + t'_{i+1} + \cdots t'_{j-1}\} \le \{t'_j\}\), for all \(1 \le i \le j-1\)
			\item \(\{t_j\} \le \{t_i + t_{i+1} + \cdots t_{j-1}\}\) if{}f \(\{t'_j\} \le 1 - \{t'_i + t'_{i+1} + \cdots t'_{j-1}\}\), for all \(1 \le i \le j-1\).
		\end{itemize}
		Due to the density of real numbers, such a choice for \(\{t'_j\}\) is always possible.
		This way, we pick the fractional values \(\{t'_1\}, \{t'_2\}\) and so on, one by one. \(\blacksquare\)
	\end{proof}
	We now prove \(\Ll(\bar{q}, 0) = \Ll(\bar{p}, 0)\).
	Pick \((z \cdot a)\,w \in \Ll(\bar{q},0)\) for some \(z \in \region(q)\).
	Using the claim, we can manufacture a word \(w'\) such that \((x \cdot a)\,w'\) has a run with the same sequence of transitions as \((z \cdot a)\,w\), and hence \((x \cdot a)\,w' \in \Ll(\bar{q},0)\).
	This implies \(w' \in \Ll(q, x)\).
	Therefore, \(w' \in \Ll(p, 0)\).
	Now, consider the word \(w\).
	Because of the properties of \(w\) and \(w'\), every \(1\)-DTA induces the same sequence of transitions in their respective runs.
	Since \(w' \in \Ll(p, 0)\), we get \(w \in \Ll(p, 0)\) too.
	By construction of the transition from \(\bar{p}\) to \(p\), we deduce that \((z \cdot a)\,w \in \Ll(\bar{p}, 0)\).
	Hence \(\Ll(\bar{q}, 0) \subseteq \Ll(\bar{p},0)\).

	For the converse, pick \((z\cdot a)\,w \in \Ll(\bar{p}, 0)\).
	Consider the same \((x \cdot a)\,w'\) as discussed in the above paragraph (every \(1\)-DTA follows the same sequence of transitions while reading both \((z \cdot a)\,w\) and \((x \cdot a)\,w'\)).
	We have \((x\cdot a)\,w' \in \Ll(\bar{p},0)\), from which we conclude \(w' \in \Ll(p, 0) = \Ll(q, x)\).
	Therefore, \((x \cdot a)\,w' \in \Ll(\bar{q},0)\).
	Once again, since \((z \cdot a)\,w\) and \((x \cdot a)\,w'\) are distinguishable by no \(1\)-DTA, we can conclude that \((z \cdot a)\,w \in \Ll(\bar{q}, 0)\).
	This proves \(\Ll(\bar{p}, 0) \subseteq \Ll(\bar{q}, 0)\).
	\qed
\end{proof}
Next we show that badness propagates along transitions that do not reset the clock.
\begin{lemma}\label{lem:bad-to-next-transition}
	Let \(q\) be a bad state in a strict acceptor.
	For every outgoing transition \(\theta = (q, q', a, \phi, r)\) from \(q\) such that
	\begin{itemize}
		\item \(r = 1\) (clock is not reset) and
		\item \(\phi\) is either an interval of the form \((j, j+1)\) for \(j \ge i\) (constant used in \(\phi\) is at least \(i\)) or, it is ``\(K<x\)'' for \(K\) the maximum constant of the acceptor.
	\end{itemize}
	the target state \(q'\) is bad too.
\end{lemma}
\begin{proof}
	Since the acceptor is strict and \(q\) is bad either \(region(q)=[K+1]_{\equiv^{K}}\) or, \(region(q)=(i,i+1)\) for some \(0\leq i<K\).
	In the first case, \(region(q')=[K+1]_{\equiv^{K}}\).
	Thus, by Remark~\ref{remark:bad_states} \(q'\) is bad.
	We now cosider the case \(region(q)=(i,i+1)\).
	There exists an \(x \in (i, i+1)\) such that \(\Ll(q,x)\) is \(1\)-DTA recognizable.
	From Lemma~\ref{lem:bad-state-same-language-for-all-x}, \(\Ll(q, i+0.5) = \Ll(q, x)\).
	Let \(\Cc\) be the \(1\)-DTA recognizing \(\Ll(q, i + 0.5)\), and let \(p\) be its initial state.
	W.l.o.g we can assume that \(\Cc\) is a strict acceptor.
	Therefore, we have \(\Ll(q, i+0.5) = \Ll(p,0)\).

	Now, let us look at a transition \(\theta = (q, q', a, \phi, r)\) satisfying the hypotheses mentioned in the lemma.
	The difference \(j - i\) is a non-negative integer.
	Consider a delay-action pair \((j-i\cdot a)\).
	In the \(1\)-DTA \(\Aa\), reading \((j-i \cdot a)\) from \((q, i+0.5)\), gives a clock value \(j+0.5\) and triggers the transition \(\theta\).
	The configuration reached on reading \((j-i \cdot a)\) is \((q', j+0.5)\).

	On the other hand, in the \(1\)-DTA \(\Cc\), reading \((j-i \cdot a)\) from \((p, 0)\) triggers a transition which resets the clock, as \(j-i\) is an integer and \(\Cc\) is a strict acceptor.
	This shows that \((j-i \cdot a)^{-1} \Ll(p, 0)\) is \(1\)-DTA recognizable.
	We deduce that \(\Ll(q', j+0.5)\) is \(1\)-DTA recognizable, hence that \(q'\) is bad by definition.
	\qed
\end{proof}

\subsection{Characterization of Bad States in Strict Acceptors}
\label{app:characterization}

\begin{proposition}
	\label{prop:bad-state-all-residuals-same}
	Let \(q\) be a bad state in a strict acceptor.
	For all \(x \in \region(q)\), for all \(a \in \Sigma\), and for all \(t, t' \ge 0\), we have: \((t \cdot a)^{-1} \Ll(q, x) = (t' \cdot a)^{-1} \Ll(q, x)\).
\end{proposition}
\begin{proof}
	Fix an arbitrary \(x \in \region(q)\) for the rest of the proof.
	Since the acceptor is strict its reset function has range included in \(\mathbb{R}_{\geq 0}\setminus\mathbb{N}_{>0}\).
	Therefore, if a state \(q\) is \emph{bad} then \(\region(q)\) is either \([K+\frac{1}{2}]_{\equiv^{K}}\) for \(K\) the constant of the acceptor or it is of the form \((i, i+1)\) for some \(i < K\).

	For a subset \(S \subseteq \mathbb{R}\), let us write \(\eta(S)\) for the predicate:
	\begin{align*}
		\eta(S):  \forall t, t' \in S, (t \cdot a)^{-1} \Ll(q, x) = (t' \cdot a)^{-1} \Ll(q, x)
	\end{align*}
	To prove the lemma, we need to show \(\eta(\mathbb{R}_{\ge 0})\).
	We start with the case \(\region(q)=[K+\frac{1}{2}]_{\equiv^{K}}\).
	Then all \((t\cdot a)\) for \(t\in \mathbb{R}_{\geq 0}\) take the same transition from \(\Ll(q, x)\) (the one with guard ``\(K<x\)'').
	Thus, they reach the same state with the same region.
	If this transition is resseting then clearly all residuals \((t \cdot a)^{-1} \Ll(q, x)\) are the same.
	Otherwise, they all reach a bad state (by Remark~\ref{remark:bad_states}).
	We conclude that all residuals are the same by Lemma~\ref{lem:bad-state-same-language-for-all-x}.

	We now consider the case \(\region(q)=(i, i+1)\).
	Here is a direct consequence of the definition of \(\eta(S)\):
	\begin{itemize}
		\item for \(S, S' \subseteq \mathbb{R}\) such that \(S \cap S' \neq \emptyset\), \(\eta(S)\) and \(\eta(S')\) implies \(\eta(S \cup S')\)
	\end{itemize}
	Due to the above observation, it is sufficient to show \(\eta([i,i+1])\) for all \(i \in \mathbb{N}\) -- notice that \([i, i+1]\) and \([i+1, i+2]\) intersect at \(i+1\).
	Below, we give the proof for \(\eta([0,1])\).
	The proof for any \([i,i+1]\) follows a similar reasoning.
	Let us now show \(\eta([0,1])\).
	Pick \(\varepsilon \in \mathbb{R}\) such that \(0 < \varepsilon < \min(\{x\}, 1 - \{x\})\).
	Once again, by the consequence of the definition, it is sufficient to show:
	\begin{align*}
		 & \eta(~[0, 1-\{x\})~)\enspace, &  & \eta(~(1-\{x\} - \varepsilon, 1 - \{x\} + \varepsilon)~)\enspace\text{, and} &  & \eta(~(1-\{x\}, 1]~)\enspace .
	\end{align*}

	\paragraph{To show \(\eta(~[0, 1-\{x\})~)\).}
	Let \(t, t' \in [0, 1 - \{x\})\).
	Then, \(x + t \equiv x + t'\), and there exists a transition \(\theta: (q, q', a, \phi, r)\) which is triggered on reading \((t \cdot a)\) or \((t' \cdot a)\) from \((q, x)\).
	If \(r = 0\) (clock is reset), then \((t \cdot a)^{-1}\Ll(q, x) = \Ll(q',0) = (t' \cdot a)^{-1} \Ll(q,x)\).
	Else, from Lemma~\ref{lem:bad-to-next-transition}, state \(q'\) is bad.
	Hence, from Lemma~\ref{lem:bad-state-same-language-for-all-x}, \(\Ll(q', x+t) = \Ll(q', x+t')\), which proves the required conclusion, since \( (t \cdot a)^{-1}\Ll(q, x) = \Ll(q', x + t)\) and \((t' \cdot a)^{-1} \Ll(q,x) = \Ll(q, x + t')\).

	\paragraph{To show \(\eta(~(1-\{x\} - \varepsilon, 1-\{x\} + \varepsilon)~)\).}
	Let \(\hat{t} = 1 - \{x\}\).
	Then \(x + \hat{t}\) is an integer.
	The transition reading \((\hat{t} \cdot a)\) from \((q, x)\) resets the clock (because we assume \(\Aa\) to be strict).
	We conclude \((\hat{t} \cdot a)^{-1} \Ll(q,x)\) is \(1\)-DTA recognizable.
	Now, consider the transition reading \((\hat{t} \cdot a)\) from \((p,0)\) in \(\Cc_{(q,x)}\) (the \(1\)-DTA recognizing \(\Ll(q,x)\)).
	Let this transition be \((p, p', q, \phi', r')\).
	\begin{itemize}
		\item If \(r' = 0\) (clock is reset), then for all \(t' \in (1-\{x\} -\varepsilon, 1-\{x\} + \varepsilon)\), we have \(t' \equiv \hat{t}\).
		      So, all \((t' \cdot a)\) take the same transition in \(\Cc_{(q,x)}\).
		      This proves \((t' \cdot a)^{-1}\Ll(p,0) = (\hat{t} \cdot a)^{-1}\Ll(p,0)\).
		      Since \(\Ll(p,0) = \Ll(q, x)\), we achieve the required conclusion.
		\item Suppose \(r' = 1\) (clock is not reset).
		      Since \((\hat{t} \cdot a)^{-1} \Ll(q,x)\) is \(1\)-DTA recognizable (see above), we also have \((\hat{t} \cdot a)^{-1} \Ll(p,0)\) is \(1\)-DTA recognizable.
		      But the transition does not reset the clock.
		      Hence \(p'\) is a bad state of \(\Cc_{(q,x)}\).
		      By Lemma~\ref{lem:bad-state-same-language-for-all-x}, we deduce \(\Ll(p', t') = \Ll(p', \hat{t})\) for all \(t' \in (1-\{x\} -\varepsilon, 1-\{x\} + \varepsilon)\).
		      This proves \((t' \cdot a)^{-1}\Ll(p,0) = (\hat{t} \cdot a)^{-1}\Ll(p,0)\), and once again, as \(\Ll(p,0) = \Ll(q, x)\), we get the required conclusion.
	\end{itemize}

	\paragraph{To show \(\eta(~(1- \{x\}, 1]~)\).}
	For any \(t, t' \in (1 - \{x\}, 1]\), we have \(x + t \equiv x + t'\).
	The rest of the reasoning is analogous to the proof of \(\eta(~[0, 1-\{x\})~)\).
	\qed
\end{proof}
\begin{corollary}\label{cor:bad-state-shape-of-language}
	Let \(q\) be a bad state in \(\Aa\).
	For each \(a \in \Sigma\) let \(\theta_a = (q, q_a, a, x=K, 0)\) be the specific transition on letter ``\(a\)'' with guard \(x = K\).
	Then, for any \(x \in \region(q)\):
	\begin{align}
		\Ll(q,x)=\bigcup_{a\in \Sigma}\{(t \cdot a)\mid t\geq 0\}\Ll(q_{a},0)\cup \{\epsilon\mid q\in F\}\enspace .
	\end{align}
\end{corollary}
\begin{proof}
	Follows from Proposition~\ref{prop:bad-state-all-residuals-same} by picking \(t = K - x\), and \(t' \ge 0\), and the observation that \((t \cdot a)^{-1} \Ll(q, x) = \Ll(q_a,0)\).
	\qed
\end{proof}

\subsection{Eliminating Bad States}\label{app:eliminating_bad_states}
Our procedure for eliminating bad states from a strict acceptor is given by Algorithm~\ref{alg:learning-rl-acceptor}.

Given a strict \(K\)-acceptor \(\Aa\), and a state \(q\) of \(\Aa\) such that \(\region(q) \neq \{0\}\), we define the strict acceptor \(\Aa_q\) obtained by modifying \(\Aa\) as follows:
\begin{itemize}
	\item Replace every transition \((q', q, a, \phi, 1)\) with \(q' \neq q\) in \(\Aa\) (there are no transitions incoming to \(q\) where clock is reset) by the transition \((q', q, a, \phi, 0)\) (clock is necessarily reset).
	\item Remove all transitions of source \(q\).
	\item Add a transition \((q, q_a, a, 0 \le x, 0)\), where \(q_a\) is the state corresponding to the transition \((q, q_a, a, x = K, 0)\) in \(\Aa\).
	\item State \(p\) is accepting (initial) in \(\Aa_q\) if{}f \(p\) is an accepting (initial) state of \(\Aa\).
\end{itemize}
Next claim is used for proving correcteness of Algorithm~\ref{alg:learning-rl-acceptor}.

\begin{proposition}
\label{prop:correctness}
	State \(q\) is bad if{}f \(\Ll(\Aa) = \Ll(\Aa_q)\).
\end{proposition}
\begin{proof}
	To avoid confusion, let \(\bar{q}\) denote the state \(q\) in \(\Aa_q\).
	Suppose \(q\) is bad.
	Then there is a \(1\)-DTA \(\Cc\) recognizing \(\Ll(q, x)\) for every \(x \in \region(q)\).
	Think of \(\bar{q}\) as the initial state of this \(1\)-DTA.
	Therefore orienting every transition \((q', q, a, \phi, 1)\) to \(\bar{q}\) with the clock reset preserves the language.
	Now, Corollary~\ref{cor:bad-state-shape-of-language} gives the shape of \(\Ll(q, x) = \Ll(\bar{q}, 0)\), which essentially says that for each \(a \in \Sigma\), the automaton can read any delay, and go to \(q_a\) after resetting the clock.
	Hence, adding the transition \((\bar{q}, q_a, a, x \ge 0, 0)\) (and removing any outgoing transition of \(\bar{q}\) in the \(1\)-DTA \(\Cc\)) preserves the language.
	The combined effect of these modifications results in \(1\)-DTA that coincides with \(\Aa_q\).
	Hence \(\Ll(\Aa) = \Ll(\Aa_q)\).

	Suppose \(\Ll(\Aa) = \Ll(\Aa_q)\).
	Consider a word \(u\) for which the run of \(\Aa\) on \(u\) ends at configuration \((q, x)\).
	Since \(\Ll(\Aa) = \Ll(\Aa_q)\), we have \(u^{-1} \Ll(\Aa) = u^{-1} \Ll(\Aa_q)\).
	Therefore \(\Ll(q,x) = \Ll(\bar{q},0)\).
	This shows that \(\Ll_{\Aa}(q,x)\) is \(1\)-DTA recognizable, and hence is a bad state.
	\qed
\end{proof}
We now prove correctness of Algorithm~\ref{alg:learning-rl-acceptor}.
By Proposition~\ref{prop:correctness} the modification \(\Aa_q\) preserves the language if{}f \(q\) is a bad state.
In this case, the number of bad states strictly decreases: the bad state \(q\) in \(\Aa\) is replaced by a state of region zero in \(\Aa_q\) while all other states remain unchanged.
Therefore, each iteration of the loop in Algorithm~\ref{alg:learning-rl-acceptor} strictly decreases the number of bad states.
Consequently, the algorithm terminates with a language-equivalent $1$-DTA that has no bad states.
Finally, we can conclude using Lemma~\ref{lemma:no-bad-rl} that the output $1$-DTA is a \(\rl\)-acceptor.

Deciding equivalence between \(1\)-DTAs \(\Aa\) and \(\Bb\) can be done as follows.
Check \(\Ll(\Aa) \subseteq \Ll(\Bb)\) and \(\Ll(\Bb) \subseteq \Ll(\Aa)\).
Observe that \(\Ll(\Aa) \subseteq \Ll(\Bb)\) if{}f \(\Ll(\Aa) \cap \Ll(\Bb)^c = \emptyset\).
Complementing \(\Bb\) can simply be done by interchanging accepting and rejecting states.
Emptiness of \(1\)-DTA (in fact, any timed automaton) is decidable (PSPACE-complete).

\subsection{Proof of the Second Item of Theorem~\ref{thm:computing-rl-acceptor}}
\label{app:proof}
By starting from a \(1\)-DTA of constant \(K_{L}\) the constant of \(\Aa\) is also \(K_{L}\) (see Lemmas~\ref{lemma:acceptorbis}, \ref{lemma:stepone} and \ref{lemma:steptwo}).
Thus, the automaton returned by Algorithm~\ref{alg:learning-rl-acceptor} is a strict acceptor of constant \(K_{L}\).
Hence, the range of its reset function is inlcuded in \([0,K_L]\setminus\mathbb{N}_{>0}\).
Since its reset function is \(\rl\) we conclude that \(\forall u\in \mathbb{T}\Sigma^{*}\), \(\rl(u)\in [0,K_L]\setminus\mathbb{N}_{>0}\).

\section{Proofs from Section~\ref{sec:syntactic}}
Proof for Lemma~\ref{lemma:steptwo} is in Appendix~\ref{app:constructing-rl-acceptor}.
\begin{proposition}
	\label{prop:bad_state_charact}
	Let \(\Aa\) be a strict acceptor such that \(L=L(\Aa)\).
	For every \( u \in \mathbb{T}\Sigma^{*}\) we have \(\rl(u)\neq R^{\Aa}(u) \) iff \(\Aa(u)\) is a bad state.
\end{proposition}
\begin{proof}
	If \(\Aa(u)\) is a bad state then by Definition~\ref{def:bad_state} \( R^{\Aa}(u)\neq 0\) and \(u^{-1}L\) is a \(1\)-DTA recognizable language.
	Thus \(\rl(u)=0\) and \(\rl(u)\neq R^{\Aa}(u) \).
	For the converse implication assume \(\rl(u)\neq R^{\Aa}(u) \).
	Let \(u=(t_1\cdot a_{1})\dots (t_n \cdot a_{n})\) and for \(j\in \{0, \dots, n\}\) let \(u_j=(t_1\cdot a_{1})\dots (t_j \cdot a_{n})\) be the prefix of length \(j\) of \(u\).
	Let \(u_{m}\) be the logest prefix of \(u\) such that \(\rl(u_{m})= R^{\Aa}(u_{m}) \) and \(\rl(u_{m+1})\neq R^{\Aa}(u_{m+1}) \) (such a prefix exists or otherwise we would have \(\rl(u)= R^{\Aa}(u) \)).
	Necessarily, \( R^{\Aa}(u_{m+1})= R^{\Aa}(u_{m})+t_{m+1}\) and \( \rl(u_{m+1})= 0\).
	Thus, \( R^{\Aa}(u_{m+1})>0\) and \(u_{m+1}^{-1}L\) is a \(1\)-DTA recognizable language and so, by definition, \(\Aa(u_{m+1})\) is a bad state.
	We will show that for every \(m+1 \leq j\leq n\) we have \(R^{\Aa}(u_j) >0\).
	Hence, for every \(m+1 \leq j\leq n-1\) the transition from \(\Aa(u_j)\) to \(\Aa(u_{j+1})\) is non ressetting.
	We will then use that bad states propagate along non ressetting transitions (Lemma~\ref{lem:bad-to-next-transition}) to deduce that if \(\Aa(u_j)\) is bad then \(\Aa(u_{j+1})\) is also bad.
	Since \(\Aa(u_{m+1})\) is bad, we can thus conclude that \(\Aa(u)\) is bad.
	To show that for every \(m+1 \leq j\leq n\) we have \(R^{\Aa}(u_j) >0\) it suffices to show that \(\rl(u_j)\neq R^{\Aa}(u_j) \).
	Indeed, we have \(0\leq \rl(u_j)\) and \(\rl(u_j)\leq R^{\Aa}(u_j) \) (by the claim below).
	Thus \(\rl(u_j)\neq R^{\Aa}(u_j) \) implies \(R^{\Aa}(u_j) >0\).
	It is true that \(\rl(u_j)\neq R^{\Aa}(u_j) \) for \(m+1 \leq j\leq n\) or, otherwise we would get a contradiction to the definition of \(u_m\).
	\begin{claim}
		Let \(L=L(\Aa)\), where \(\Aa\) is a \(1\)-DTA.
		For every \(u\in \mathbb{T}\Sigma^{*}\) we have \(\rl(u) \leq R^{\Aa}(u)\).
	\end{claim}
	\begin{proof}
		By induction on the length of words.
		It is true for \(\epsilon\).
		Assume \(\rl(u) \leq R^{\Aa}(u)\).
		If \(\rl(u(t\cdot a))=0\) then clearly \(\rl(u(t\cdot a)) \leq R^{\Aa}(u(t\cdot a))\).
		Otherwise, \((u(t\cdot a))^{-1}L\) is not a \(1\)-DTA recognizable language, thus \(R^{\Aa}(u(t \cdot a)) \neq 0\).
    Hence, \(\rl(u(t \cdot a))=\rl(u) +t \leq R^{\Aa}(u) + t= R^{\Aa}(u(t \cdot a))\) by the induction hypothesis. \(\blacksquare\)
	\end{proof}
	\qed
\end{proof}
\begin{proposition}
	\label{prop:comparison_cL_strict}
	Let \(\Aa\) be a strict acceptor such that \(L=L(\Aa)\).
	We have \(\forall u \in \mathbb{T}\Sigma^{*},~\rl(u)= R^{\Aa}(u)\vee \rl(u)=0\).
\end{proposition}
\begin{proof}
	We need to show that \(\rl(u)\neq R^{\Aa}(u)\implies \rl(u)=0\).
	Assume \(\rl(u)\neq R^{\Aa}(u)\).
	Then by Proposition~\ref{prop:bad_state_charact} \(\Aa(u)\) is a bad state, thus by definition, \(u^{-1}L\) is a \(1\)-DTA recognizable language, thus \(\rl(u)=0\).
	\qed
\end{proof}

\begin{lemma}
	\label{lemma:finite_index}
	Let \(\Aa\) be a strict acceptor such that \(L=L(\Aa)\).
	The equivalence \(\approx^\Aa\) refines \(\approx^{L}\).
\end{lemma}
\begin{proof}
	Let \(u, v \in \HI\) with \(u \approx^{\Aa} v\) i.e., \(\Aa(u)=\Aa(v)\).
	Since \(\Aa\) is an acceptor we have \(R^{\Aa}(u) \equiv^{K}R^{\Aa}(v)\) for \(K\) the maximum constant of \(\Aa\).
	Since the two values \(R^{\Aa}(u) \) and \(R^{\Aa}(v) \) are half-integrals, \(R^{\Aa}(u) \equiv^{K}R^{\Aa}(v)\) implies that either \(R^{\Aa}(u)=R^{\Aa}(v)\) or, both values are strictly above \(K\).
	In the first case, \(u^{-1}L=v^{-1}L\).
	Hence, \(\rl(u)=0\iff \rl(v)=0\).
	Then, using Proposition~\ref{prop:comparison_cL_strict}, it is an easy exercise to conclude \(\rl(u)=\rl(v)\).
	In the second case where \(R^{\Aa}(u)>K \) and \(R^{\Aa}(v) >K\) by Remark~\ref{remark:bad_states} the states \(R^{\Aa}(u) \) and \(R^{\Aa}(v)\) are bad.
	Thus, \(\rl(u)= \rl(v)=0\) and by Lemma~\ref{lem:bad-state-same-language-for-all-x} \(u^{-1}L=v^{-1}L\).
	Therefore, \(u \approx^{L} v\).
	\qed
\end{proof}

\begin{propositionstar}[\ref{prop:mh}]
	If \(L\) is \(1\)-DTA recognizable language then \((\rl,\approx^{L})\) is \(L\)-preserving, monotonic and has finite index and \(K_{(\rl,\approx^{L})}=K_{L}\).
\end{propositionstar}
\begin{proof}
	We start with condition {\upshape(\itshape a\upshape)} and {\upshape(\itshape b\upshape)} of monotonicity which hold for any language \(L\) (not only the \(1\)-DTA recognizable).
	Condition {\upshape(\itshape a\upshape)} holds by definition of \(\approx^{L}\).
	For condition {\upshape(\itshape b\upshape)} let \(u\approx^{L}v\) and \(t\in \mathbb{R}_{\geq 0}\) and \( a\in \Sigma\).
	We have \(\rl(u)= \rl(v) \) and \(u^{-1}L=v^{-1}L\).
	Thus, \(\rl(u)+t= \rl(v)+t \) and \((u(t\cdot a))^{-1}L=(v(t\cdot a))^{-1}L\).
	Either \((u(t\cdot a))^{-1}L\) is \(1\)-DTA recognizable and in this case \(\rl(u(t\cdot a))=\rl(v(t\cdot a))=0\) or, \((u(t\cdot a))^{-1}L\) is not \(1\)-DTA recognizable and so we have \(\rl(u(t\cdot a))=\rl(u) +t =\rl(v)+t=\rl(v(t \cdot a))\).
	Hence, \(u(t \cdot a)\approx^{L}v(t \cdot a)\).

	We now show \(L\)-preservation and monotonicity {\upshape(\itshape c\upshape)}.
	Let \(u\approx^{L}v\).
	Since \(u^{-1}L=v^{-1}L\) we have \(u\in L \iff \epsilon \in u^{-1}L \iff \epsilon \in v^{-1}L \iff v\in L\).
	For the second condition of \(L\)-preservation and for monotonicity {\upshape(\itshape c\upshape)} we need to show that \(u\in L\) iff \(N^{\rl}(u)\in L\) and \(\rl(u)\equiv \rl(N^{\rl}(u))\).
	We use that there exists an acceptor for \(L\) whose reset function is \(\rl\) (Theorem~\ref{thm:computing-rl-acceptor}).
	This acceptor sends \(u\) and \(N^{\rl}(u)\) to the same state, thus \(u\in L\) iff \(N^{\rl}(u)\in L\) and \(\rl(u)\equiv \rl(N^{\rl}(u))\).

	Finite index for \(\approx^{L}\) is a consequence of Lemma~\ref{lemma:steptwo} and Lemma~\ref{lemma:finite_index}: By Lemma~\ref{lemma:steptwo} there is a strict acceptor for \(L\) and by Lemma~\ref{lemma:finite_index} its equivalence refines \(\approx^{L}\).
	Since this equivalence has finite index so does \(\approx^{L}\).

	We now show that \(K_{(\rl,\approx^{L})}\) is finite and equal to \(K_L\).
	We consider the proposition \(C\) w.r.t the profile \((\rl,\approx^{L})\) and show that \(C(K_{L})\) holds, thus \(K_{(\rl,\approx^{L})}\leq K_{L}< \infty\).
	Then, from Proposition~\ref{prop:monotone-to-automaton} the acceptor defined from \((\rl,\approx^{L})\) accepts \(L\).
	Hence, by definition of \(K_{L}\) we have \( K_{L}\leq K_{(\rl,\approx^{L})}\).
	Thus, \(K_{(\rl,\approx^{L})}= K_{L}\).
	It remains to show that \(C(K_{L})\) holds.
	Let \(u\in \mathbb{T}\Sigma^{*}\), \(a\in \Sigma\) and \(n\in \mathbb{N}\) such that \(\rl(u)+\frac{n}{2}>K_{L}\).
	We need to show that \(u(\frac{n}{2}\cdot a)\approx^{L} u(K_{L}+\frac{1}{2}\cdot a)\).
	For that consider an acceptor for \(L\) whose reset function is \(\rl\) and whose maximum constant in \(K_{L}\) (it exists by Theorem~\ref{thm:computing-rl-acceptor}).
	Since \(\rl(u)+\frac{n}{2}\) and \(\rl(u)+K_{L}+\frac{1}{2}\) are both striclty above the maximum constant \(K_L\) of the acceptor, the last transitions on reading \(u(\frac{n}{2}\cdot a)\) and \(u(K_{L}+\frac{1}{2}\cdot a)\) are the same and the two words end up in the same state of the acceptor.
	Moreover, by the second item of Theorem~\ref{thm:computing-rl-acceptor}, \(\rl\colon \mathbb{T}\Sigma^{*}\to [0,K_L]\).
	Thus, we must have \(\rl(u(\frac{n}{2}\cdot a))=\rl(u(K_{L}+\frac{1}{2}\cdot a))=0\).
	Hence, the residuals of the two words are the same and so \(u(\frac{n}{2}\cdot a)\approx^{L} u(K_{L}+\frac{1}{2}\cdot a)\).
	\qed
\end{proof}
\begin{propositionstar}[\ref{prop:rl_is_strict}]
	If \(L\) is a \(1\)-DTA recognizable language then \(\Aa_{(\rl, \approx^L)}\) is a strict \(K_{L}\)-acceptor with reset function \(\rl\).
	Moreover, no strict acceptor for \(L\) has fewer states than \(\Aa_{(\rl, \approx^L)}\).
\end{propositionstar}
\begin{proof}
	The reset function of \(\Aa_{(\rl, \approx^L)}\) is \(\rl\) (Proposition~\ref{prop:monotone-to-automaton}) and its constant is \(K_{(\rl,\approx^{L})}=K_{L}\) (Proposition~\ref{prop:mh}).
	Since \(\forall w\in \mathbb{T}\Sigma^{*}\), \(\rl(w)\in[0,K_L]\setminus\mathbb{N}_{>0}\) (Theorem~\ref{thm:computing-rl-acceptor}) \(\Aa_{(\rl, \approx^L)}\) is strict.
	Minimality is a consequence of Lemma~\ref{lemma:finite_index}.
	\qed
\end{proof}

\subsection{Deciding \(\mathcal{L}(q,\frac{n}{2})=\mathcal{L}(q',\frac{n}{2})\)}
\label{app:deciding_language_equality_for states}
Let \(\Aa\) be a acceptor and \(q\) and \(q'\) be two states in \(\Aa\) such that \(\region(q)=\region(q')\).
Let \(\frac{n}{2}\) be the unique half-integral value in \(\region(q)\).
We modify \(\Aa\) by adding two transitions:
\begin{itemize}
	\item add \((\bar{q}, q, a, \phi_{K}(\frac{n}{2}), 1)\),
	\item add \((\bar{q}', q', a, \phi_{K}(\frac{n}{2}), 1)\).
\end{itemize}

\begin{lemma}
	We have \(\mathcal{L}(q,\frac{n}{2})=\mathcal{L}(q',\frac{n}{2})\) iff \(\mathcal{L}(\bar{q},0)=\mathcal{L}(\bar{q}',0)\).
\end{lemma}
\begin{proof}
	Using the claim in the proof of Lemma~\ref{lem:bad-state-same-language-for-all-x} for every timed word \((y \cdot a)\,w \) with \(y \in \region(q)\) we can manufacture a word \(w'\) such that \((y \cdot a)\,w\) and \((\frac{n}{2} \cdot a)\,w'\) follow the same sequence of transitions in any \(1\)-DTA.
	Therefore, \((y \cdot a)\,w\in \Ll(\bar{q}, 0) \iff (\frac{n}{2} \cdot a)\,w'\in \Ll(\bar{q}, 0)\).
	Since \((\frac{n}{2} \cdot a)\,w'\in \Ll(\bar{q}, 0) \iff w'\in \Ll(q, \frac{n}{2})\) we find that \((y \cdot a)\,w\in \Ll(\bar{q}, 0) \iff w'\in \Ll(q, \frac{n}{2})\).
	Similarly, we have that \((y \cdot a)\,w\in \Ll(\bar{q}', 0) \iff w'\in \Ll(q', \frac{n}{2})\).
	Hence, \(\mathcal{L}(q,\frac{n}{2})=\mathcal{L}(q',\frac{n}{2})\) implies \(\mathcal{L}(\bar{q},0)=\mathcal{L}(\bar{q}',0)\).

	For the converse direction we use that for every word \(w\in \mathbb{T}\Sigma^{*}\) we have \(w\in \Ll(q, \frac{n}{2})\iff (\frac{n}{2} \cdot a)\,w\in \Ll(\bar{q}, 0)\) and that the analogue holds with \(q'\) and \(\bar{q}'\).
	Hence, \(\mathcal{L}(\bar{q},0)=\mathcal{L}(\bar{q}',0)\) implies \(\mathcal{L}(q,\frac{n}{2})=\mathcal{L}(q',\frac{n}{2})\).
	\qed
\end{proof}

\begin{lemmastar}[\ref{lemma:computing_canonical_acc}]
	\(C_{\Bb}(K_L)\) holds.
\end{lemmastar}
\begin{proof}
	Let \((q,q',a,\phi_K(t),r)\in T\) such that \(K_L<t\).
	We will show that \(q'= q_a\).
	Let \(q=\Bb(u)\) and \(q'=\Bb(u(t'\cdot a))\).
	We have \(\rl(u)+t'>K_L\) and \(\rl(u)+K+\frac{1}{2}>K_L\).
	By Proposition~\ref{prop:mh} \(C(K_{L})\) holds for \(C\) taken w.r.t the profile \((\rl,\approx^{L})\).
	Hence, \(u(t'\cdot a)\approx^{L}u(K+\frac{1}{2}\cdot a)\).
	Since \(\approx^{L}\) and \(\approx^{\Bb}\) coincide we find that \(\Bb(u(t'\cdot a))= \Bb(u(K+\frac{1}{2}\cdot a))\), that is \(q'=q_a\).\qed
\end{proof}

\section{Examples for Section~\ref{sec:syntactic}}
\label{app:myhill-nerode}
The next example shows that \(\Aa_{(\rl, \approx^L)}\) is not necessarily minimal among all acceptors.
\begin{example}
	\label{example:counterexample_minimality}
	Consider the languages \(L=D\cup\{(t_{1}\cdot d)(t_{2}\cdot c)\mid t_{1}\in (1,2) \text{~and~}t_{1}+t_{2}=2\}\) where \(D=\{(1\cdot a)(t_{1}\cdot b)(t_{2}\cdot c)\mid t_{1}\in (0,1) \text{~and~}t_{1}+t_{2}=1\}\).
	This language is accepted by the automata of \figurename~\ref{fig:counterexample_minimality}.
	The possible residual languages of half-integral words are five: \(\emptyset\), \(L\), \((1\cdot a)^{-1}L\), \((1\cdot a)(\frac{1}{2}\cdot b)^{-1}L\) and \(\{\epsilon\}\).
	Next we show that any acceptor for \(L\) with \(R^{\Aa}(1\cdot a)=0\) must have at least six states (including the sink state).
	Indeed, five states are already required to represent the five distinct residuals.
	Moreover, an additional state is needed to separate the words \(u=(1\cdot a)(\frac{1}{2}\cdot b)\) and \(v=(1+\frac{1}{2}\cdot b)\) which have the same residual but different clock values.
	Their common residual is not \(1\)-DTA recognizable and since \(R^{\Aa}(1\cdot a)=0\) we find \(R^{\Aa}(u)=\frac{1}{2}\) while \(R^{\Aa}(v)=1 +\frac{1}{2}\).
	Thus, \(u\) and \(v\) cannot be sent to the same state.
	In particular, since \(\Aa_{(\rl, \approx^L)}\) is a strict acceptor (hence \(\rl(1\cdot a)=0\)) it has at least six states.
	The top automaton in \figurename~\ref{fig:counterexample_minimality} accepts \(L\) with five states (including the sink).
	Therefore, \(\Aa_{(\rl, \approx^L)}\) is not minimal.
\end{example}

\begin{figure}[ht]
	\centering
	\begin{tikzpicture}[state/.style={circle, draw, minimum size=18pt, inner sep=1pt}, every edge/.append style={font=\scriptsize}] %
		\node[state, initial] (q0) at (0,0) {};
		\node[state] (q1) at (3,0) {};
		\node[state] (q2) at (6,0) {};
		\node[state,accepting] (f) at (9,0) {};
		\begin{scope}[->, >=stealth, auto]
			\draw (q0) edge node {\(a,x=1,1\)} (q1);
			\draw (q1) edge node [above] {\(b, 1<x<2,1\)} (q2);
			\draw (q2) edge node {\(c,x=2,1\)} (f);
			\draw (q0) edge [bend left] node {\(d,1<x<2,1\)} (q2);
		\end{scope}
	\end{tikzpicture}
	\\[20pt]
	\begin{tikzpicture}[state/.style={circle, draw, minimum size=18pt, inner sep=1pt}, every edge/.append style={font=\scriptsize}] %
		\node[state, initial] (q0) at (0,0) {};
		\node[state] (q1) at (3,0) {};
		\node[state] (q2) at (6,0) {};
		\node[state,accepting] (f) at (9,0) {};
		\node[state] (q2') at (4.5,1) {};
		\begin{scope}[->, >=stealth, dashed, auto]
			\draw (q0) edge node {\(a,x=1,0\)} (q1);
		\end{scope}
		\begin{scope}[->, >=stealth, auto]
			\draw (q1) edge node [above] {\(b,0<x<1,1\)} (q2);
			\draw (q2) edge node {\(c,x=1,1\)} (f);
			\draw (q0) edge [bend left] node {\(d,1<x<2,1\)} (q2');
			\draw (q2') edge [bend left] node {\(c,x=2,1\)} (f);
		\end{scope}
	\end{tikzpicture}
	\caption{Top: minimal acceptor for \(L=\{(1\cdot a)(t_{1}\cdot b)(t_{2}\cdot c)\mid t_{1}\in (0,1) \text{~and~}t_{1}+t_{2}=1\}\cup\{(t_{1}\cdot d)(t_{2}\cdot c)\mid t_{1}\in (1,2) \text{~and~}t_{1}+t_{2}=2\}\).
		Bottom: \(\Aa_{(\rl, \approx^L)}\).
		This example shows that \(\Aa_{(\rl, \approx^L)}\) is not necessarily minimal among all acceptors.
		For readability, sink states are omitted from the figure, though they are always assumed to be present in our reasoning.
	}
	\label{fig:counterexample_minimality}
\end{figure}
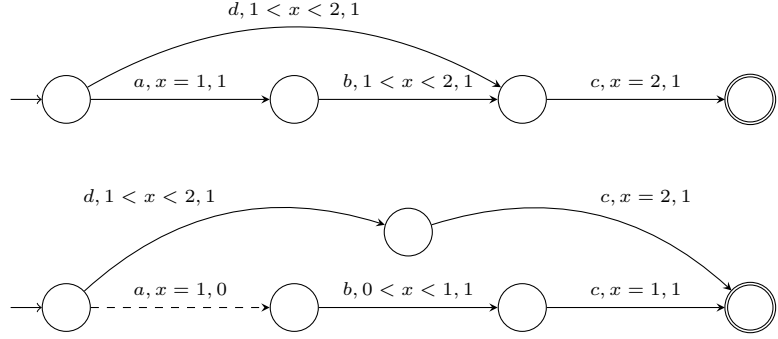

\figurename~\ref{fig:no_unique_minimal} shows two minimal acceptors for the same language with different reset functions.
Even when the timing structure is ignored, the two automata are non-isomorphic, showing that the choice of resets influences how states are merged.

\begin{figure}[hbt]
	\centering
	\begin{tikzpicture}[state/.style={circle, draw, minimum size=18pt, inner sep=1pt}, every edge/.append style={font=\scriptsize}] %
		\node[state, initial] (q0) at (0,0) {};
		\node[state] (q1) at (3,0) {};
		\node[state] (q2) at (6,0) {};
		\node[state,accepting] (f) at (9,0) {};
		\node[state] (q2') at (4.5,-1) {};
		\begin{scope}[->, auto]
			\draw (q0) edge node {\(a,x=1,1\)} (q1);
			\draw (q1) edge node {\(b, 1<x<2,1\)} (q2);
			\draw (q2) edge node {\(c,x=2,1\)} (f);
			\draw (q0) edge [bend left] node {\(d,1<x<2,1\)} (q2);
			\draw (q0) edge [bend right] node {\(e,0<x<1,1\)} (q2');
			\draw (q2') edge [bend right] node {\(c,x=1,1\)} (f);
		\end{scope}
	\end{tikzpicture}
	\\[20pt]
	\begin{tikzpicture}[state/.style={circle, draw, minimum size=18pt, inner sep=1pt}, every edge/.append style={font=\scriptsize}] %
		\node[state, initial] (q0) at (0,0) {};
		\node[state] (q1) at (3,0) {};
		\node[state] (q2) at (6,0) {};
		\node[state,accepting] (f) at (9,0) {};
		\node[state] (q2') at (4.5,1) {};
		\begin{scope}[->, >=stealth, dashed, auto]
			\draw (q0) edge node {\(a,x=1,0\)} (q1);
		\end{scope}
		\begin{scope}[->, >=stealth, auto]
			\draw (q1) edge node [above] {\(b, 0<x<1,1\)} (q2);
			\draw (q2) edge node {\(c,x=1,1\)} (f);
			\draw (q0) edge [bend left] node {\(d,1<x<2,1\)} (q2');
			\draw (q2') edge [bend left] node {\(c,x=2,1\)} (f);
			\draw (q0) edge [bend right] node {\(e,0<x<1,1\)} (q2);
		\end{scope}
	\end{tikzpicture}
	\caption{Two minimal non isomorphic acceptors with different reset functions accepting the same language (sink states are omitted).
	}
	\label{fig:no_unique_minimal}
\end{figure}
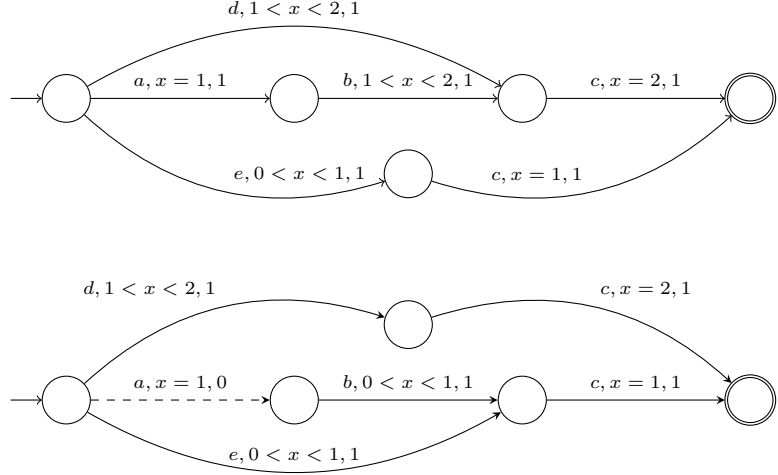
\section{Proofs from Section~\ref{sec:learning}}

\begin{propositionstar}[\ref{prop:acceptor}]
	If \(\Tt\) is closed, consistent and valid, and has constant \(K\) then \(\Aa_{{\Tt}}\) is a strict \(K\)-acceptor such that \(R^{\Aa_{{\Tt}}}:\mathbb{R}_{\geq 0}\to [0,K]\setminus\mathbb{N}_{>0}\).
	Moreover, for every \(s\in S\), \((\row(\epsilon),0) \rightsquigarrow^{s} (\row(s),r(s))\).
\end{propositionstar}
\begin{proof}
	We start by showing that \(\Aa_{{\Tt}}\) is deterministic.
	Consider two transitions \(( \row(s_1), \row(s_1'), a, g, d_1)\) and \(( \row(s_2), \row(s_2'), a, g, d_2)\) such that \(\row(s_1)=\row(s_2)\).
	By definition of \(\Aa_{{\Tt}}\) there is \((t_1 \cdot a)\in \Sigma_{K} \) such that: \(\row(s_1(t_1 \cdot a)) = \row(s_1')\), \(g = \phi_K(\mathit{r}(s_1) + t_1)\) and \(d_1 = 0\) if \(\mathit{r}(s_1')=0 \) and, \(d_1 = 1\) otherwise.
	Similarly, there is \((t_2 \cdot a)\in \Sigma_{K} \) such that: \(\row(s_2(t_2 \cdot a)) = \row(s_2')\), \(g = \phi_K(\mathit{r}(s_2) + t_2)\) and \(d_2 = 0\) if \(\mathit{r}(s_2')=0 \) and, \(d_2 = 1\) otherwise.
	Since \( \phi_K(\mathit{r}(s_1) + t_1)=\phi_K(\mathit{r}(s_2) + t_2)\) we have \(s_1 + t_1\equiv^{K}s_2 + t_2\).
	If \(s_1 + t_1\leq K\) then \(\mathit{r}(s_1) + t_1\equiv^{K}\mathit{r}(s_2) + t_2\) implies \(\mathit{r}(s_1) + t_1=\mathit{r}(s_2) + t_2\) since each bounded \(K\)-region contains exactly one half-integral value.
	Since \(\mathit{r}(s_1)=\mathit{r}(s_2)\) we thus deduce \(t_1=t_2\).
	Then by consistency (a) we deduce \(\row(s_1(t_1 \cdot a)) = \row(s_2(t_2 \cdot a)) \), thus \(\row(s_1') = \row(s_2') \).
	If \(K<s_1 + t_1\) then by consistency (b) \(\row(s_1(t_1 \cdot a)) = \row(s_1(K+\frac{1}{2} \cdot a)) \) and similarly \(\row(s_2(t_2 \cdot a)) = \row(s_2(K+\frac{1}{2} \cdot a)) \).
	By consistency (a) \(\row(s_1(K+\frac{1}{2} \cdot a)) = \row(s_2(K+\frac{1}{2} \cdot a)) \).
	Hence, \(\row(s_1(t_1 \cdot a)) = \row(s_2(t_2 \cdot a)) \), thus \(\row(s_1') = \row(s_2') \).
	Since the reset information \(d_{1}\) only depends on \([r(s_1')]_{\equiv^{K}}\) and similarly \(d_{2}\) only depends on \([r(s_2')]_{\equiv^{K}}\), \(\row(s_1') = \row(s_2') \) implies \(d_1=d_2\).

	Since for every state \(\row(s)\) we have \(region(\row(s))=[r(s)]_{\equiv^{K}}\), \(\Aa_{{\Tt}}\) is easily seen to be an acceptor.
	Finally, since \(\Tt\) is valid, \(R^{\Aa_{{\Tt}}}:\mathbb{R}_{\geq 0}\to [0,K]\setminus\mathbb{N}_{>0}\).
	The last part of the proposition is proven by an easy induction on the length of words in \(S\).\qed
\end{proof}
\section{Smart Teacher Algorithm -- Additional Material}
\begin{lemma}
	\label{lemma:more_rows}
	The next closed, consistent, valid OT obtained after processing a counterexample either has strictly more rows or, a strictly greater constant.
\end{lemma}
\begin{proof}
	Let \(w\) be the counterexample to the conjecture of \(\Tt_{w}\) and \(\Tt\) the next closed, consistent and valid OT.
	Assume for contradiction that \(\Tt_{w}\) and \(\Tt\) have the same number of distinct rows and the same constant.
	None of the actions {\upshape 1--4} was applied to obtain \(\Tt\) as such actions either increase the number of distinct rows or increase the constant.
	In particular \(E\) did not change, and since \(\Tt\) extends \(\Tt_{w}\) we deduce that the sets of rows of \(\Tt_{w}\) and \(\Tt\) are the same.
	Since a conjecture is entirely determined by the set of rows of the OT and by its constant, the conjectures \(\Aa_{{\Tt_{w}}}\) and \(\Aa_{{\Tt}}\) are the same.
	We now obtain a contradiction: Since \(w\) is a counterexample we have \(w\in L(\Aa_{{\Tt}}) \iff w\notin L\).
	Since \(N^{\rl}(w)\) is in \(\Tt\), \(\Aa_{{\Tt}}\) sends \(w\) and \(N^{\rl}(w)\) to the same state which is \(\row(N^{\rl}(w))\).
	By definition of \(\Aa_{{\Tt}}\), \(w\in L(\Aa_{{\Tt}}) \) if{}f \(\row(N^{\rl}(w))(\epsilon)(1)=1\).
	Since \(\Tt\) is computed by membership queries and the unknown langauge is \(L\) \(\row(N^{\rl}(w))(\epsilon)(1)=1\) if{}f \(N^{\rl}(w)\in L\).
	Since \(w\) and \(N^{\rl}(w)\) go to the same state of a acceptor for \(L\) with reset function \(\rl\), \(N^{\rl}(w)\in L\) if{}f \(w\in L\).
	Hence, \(w\in L(\Aa_{{\Tt}}) \iff w\in L\).
	\qed
\end{proof}

\begin{propositionstar}[\ref{prop:ST}]
	Smart Teacher algorithm terminates and returns \(\Aa_{(\rl, \approx^{L})}\).
	Moreover, the algorithm processes at most \(K_{L}+n\) OTs and the number of membership queries is \(O(mn^2 K_L |\Sigma|)\) where \(n\) is the number of equivalence classes for \(\approx^{L}\) and \(m\) the length of the longest counterexample.
\end{propositionstar}
\begin{proof}
	Each application of action~{\upshape 1} or action~{\upshape 2} adds a new row and, for \(s,s'\in S\) we have that \(\row(s) \neq \row(s')\) implies \(s \not \approx^{L}s'\).
	Thus, each application of action~{\upshape 1} or action~{\upshape 2} adds a new equivalence class for \(\approx^{L}\).
	Since there are finitely many classes for \(\approx^L\), actions~{\upshape 1} and~{\upshape 2} cannot be applied indefinitely.
	Since \(\rl\) is strictly bounded by \(K_{L}\) (second item of Theorem~\ref{thm:computing-rl-acceptor}) and since the reset assignement of \(\Tt\) corresponds to \(\rl\) once \(K=K_L\) action~{\upshape 4} can no longer be applied.
	As the next claim shows the same holds for action~{\upshape 3}.
	\begin{claim}
		If \(K=K_{L}\) then \(\Tt\) satisfies condition ~{\upshape(\itshape b\upshape)}.
	\end{claim}
	\begin{proof}
		Let \(K<r(s)+t\).
		Since \(r(s)=\rl(s)\) we have \(K<\rl(s)+t\).
		An acceptor with reset function \(\rl\) and constant \(K=K_{L}\) sends \(s(t\cdot a)\) and \(s(K+\frac{1}{2}\cdot a)\) to the same state, thus the two words are equivalent for \(\approx^{L}\).
    Therefore, they have the same row. \(\blacksquare\)
	\end{proof}
	This shows that the Learner eventually gets a closed, consistent OT from which to make a conjecture and that such OT has constant at most equal to \(K_{L}\).
	It remains to show that the Learner cannot get stuck repeating the same conjecture.
	This follows from Lemma~\ref{lemma:more_rows}: after each counterexample, the next closed, consistent and valid OT either has strictly more rows or a striclty greater constant.
	Both the number of rows and the constant are bounded by previous arguments, thus the algorithm terminates.
	It returns an acceptor isomorphic to \(\Aa_{(\rl,\approx^{L})}\).

	We now establish the complexity bounds.
	Each new OT either increases the constant or introduces a new equivalence class of \(\approx^{L}\).
	Since the constant cannot exceed \(K_{L}\) the algorithm processes at most \(K_{L}+n\) OTs.
	The number of words in \(S\) is bounded by \(nm\).
	The number of words in \(E\) is bounded by \(n\).
	The number of different extensions for each word is bounded by \(|\Sigma|(2K_L+2)\).
	Therefore, the total number of cells in the final OT is bounded by \(n^2m (1 + |\Sigma|(2K_L+2)) \).
	Hence, the number of membership queries is \(O(mn^2|\Sigma|K_L)\).
	\qed
\end{proof}

Next we give an example run of the Smart Teacher algorithm.
\begin{example}
	\label{example:sm}
	We give an example run of the Smart Teacher algorithm where the unknown language is \(L=\{(x \cdot a)(y \cdot a)\mid 0<x<1, x+y =2\}\).
	The first three OTs computed, \(\Tt_{0}\) (the initial OT), \(\Tt'_0\) and \(\Tt''_0\), are depicted in Table~\ref{tab:exampleOTs}.
	\(\Tt_0\) has constant \(K=0\).
	We obtain \(\Tt'_0\) from \(\Tt_{0}\) by increasing the constant to \(1\) (action~{\upshape 4}) because \(r( \frac{1}{2}\cdot a)=\frac{1}{2}>0\).
	Then we obtain \(\Tt_{0}''\) by applying action~{\upshape 1} (moving \(( \frac{1}{2}\cdot a)\) to \(S\)).
	We then obtain \(\Tt_{1}\) (Table~\ref{tab:exampleXXXX}) by applying action~{\upshape 1} to \(\Tt''_0\) (moving \(( \frac{1}{2}\cdot a)( 1+\frac{1}{2}\cdot a)\) to \(S\)).
	OT \(\Tt_{1}\) with \(K=1\) is inconsistent~{\upshape (b)}: \(K<r((\frac{1}{2}\cdot a))+1\) but \(\row((\frac{1}{2}\cdot a)(1\cdot a))\neq \row((\frac{1}{2}\cdot a)(1+\frac{1}{2}\cdot a))\).
	We thus increase \(K\) to \(2\).
	Now \(\Tt_{1}\) with constant \(K=2\) is closed, consistent and valid.
	Its conjecture \(\Aa_{\Tt_{1}}\) is depicted in \figurename~\ref{fig:exampleST}.
	The word \((0\cdot a)(\frac{1}{2}\cdot a)(1+\frac{1}{2}\cdot a)\) is a counterexample as it is accepted by \(\Aa_{\Tt_{1}}\) but not in \(L\).
	After adding it, along with its prefixes to \(\Tt_{1}\) we get the OT corresponding to \(\Tt_{2}\) (Table~\ref{tab:exampleXXXX}) minus the second column.
	This column is added afterwards to solve the inconsistency~{\upshape (a)}: \(\row(\epsilon)= \row((0 \cdot a))\) but \(\row((\frac{1}{2} \cdot a))\neq \row((0 \cdot a)(\frac{1}{2} \cdot a))\).
	The resulting OT \(\Tt_{2}\) (Table~\ref{tab:exampleXXXX}) is closed, consistent, valid and we have \(L(\Aa_{\Tt_{2}})=L\) (\figurename~\ref{fig:exampleST}).
\end{example}
\begin{table}
	\caption{More OTs for Example~\ref{example:sm} and \(L=\{(x \cdot a)(y \cdot a)\mid 0<x<1, x+y =2\}\).
		For \(s\in S\), we use \(s\star\) for the extensions of \(s\) (not appearing in \(S\)), we use bold when the values are the same for multiple rows.
		We use \(\star\) to mean all the extensions not appearing in \(S\).
		Both OTs are closed and consistent {\upshape (a)}.
		For \(K=1\), \(\Tt_1\) is not consistent {\upshape (b)}: \(K<r((\frac{1}{2}\cdot a))+1\), but the rows of \((\frac{1}{2}\cdot a)(1\cdot a)\) and \((\frac{1}{2}\cdot a)(1+\frac{1}{2}\cdot a)\) differ.
		For \(K=2\), both are consistent {\upshape (b)} as \(K<r((\frac{1}{2}\cdot a))+1\) does not hold.
	}
	\label{tab:exampleXXXX}
	\centering
	\begin{minipage}{0.45\textwidth}
		\centering
		\subcaption{For \(K=1\) or \(K=2\)}
		\begin{tabular}{|c|@{\hspace{2pt}}c@{\hspace{2pt}}|@{}c@{}|c|}
			\hline
			\(\Tt_{1}\)                                          & \(\epsilon\)                \\
			\hline
			\(\epsilon\)                                         & \((0,0)\)                   \\
			\((\frac{1}{2}\cdot a)\)                             & \((0,\frac{1}{2})\)
			\\
			\((\frac{1}{2}\cdot a)(1+ \frac{1}{2}\cdot a)\)      & \((1,0)\)                   \\
			\hline
			\(\epsilon \star\)                                   & \((\mathbf{0},\mathbf{0})\) \\
			\((\frac{1}{2}\cdot a)\star\)                        & \((\mathbf{0},\mathbf{0})\) \\
			\((\frac{1}{2}\cdot a)(1+ \frac{1}{2}\cdot a)\star\) & \((\mathbf{0},\mathbf{0})\) \\
			\hline
		\end{tabular}
	\end{minipage}
	\begin{minipage}{0.45\textwidth}
		\centering
		\subcaption{For \(K=2\)}
		\begin{tabular}{|c|@{\hspace{2pt}}c@{\hspace{2pt}}|@{}c@{}|c|}
			\hline
			\(\Tt_{2}\)                                               & \(\epsilon\)                & \((\frac{1}{2}\cdot a)\)    \\
			\hline
			\(\epsilon\)                                              & \((0,0)\)                   & \((0,\frac{1}{2})\)         \\
			\((\frac{1}{2}\cdot a)\)                                  & \((0,\frac{1}{2})\)         & \((0,0)\)
			\\
			\((\frac{1}{2}\cdot a)(1+ \frac{1}{2}\cdot a)\)           & \((1,0)\)                   & \((0,0)\)                   \\
			\((0\cdot a)\)                                            & \((0,0)\)                   & \((0,0)\)                   \\
			\((0\cdot a)(\frac{1}{2}\cdot a)\)                        & \((0,0)\)                   & \((0,0)\)                   \\
			\((0\cdot a)(\frac{1}{2}\cdot a)(1+ \frac{1}{2}\cdot a)\) & \((0,0)\)                   & \((0,0)\)                   \\
			\hline
			\(\star\)                                                 & \((\mathbf{0},\mathbf{0})\) & \((\mathbf{0},\mathbf{0})\) \\
			\hline
		\end{tabular}
	\end{minipage}
\end{table}
\begin{algorithm}[H]
	\nextfloat\caption{Learning a strict acceptor with a Smart Teacher}
	\KwOut{A strict  \(K_{L}\)-acceptor \(\Aa_{\Tt}\) such that \(L(\Aa_{\Tt}) = L\)}
	\((K, S, E) \gets (0, \{\epsilon\}, \{\epsilon\})\)\;
	\While{\(\mathtt{true}\)}{
		\(\Tt \gets(K,S,E,T)\)\tcp*[l]{build \(\Tt\) from \(S\), \(E\) and \(K\)}
		\If{\(\Tt\) is not closed or not consistent or not valid}{
			\uIf{ \(\exists w \in S\Sigma_K,\, \forall s \in S,\, \row(w) \neq \row(s)\)}{
				\(S \gets S \cup \{w\}\)\;
			}
			\uElseIf{\(\exists s_1, s_2 \in S,\, \exists (t \cdot a) \in \Sigma_K, \, \exists e \in E,\,\row(s_1) = \row(s_2) \land{}\)
				\(\row(s_1(t \cdot a))(e) \neq \row(s_2(t \cdot a))(e)\)}{
				\(E \gets E \cup \{(t \cdot a)e\}\)\;
			}
			\Else(\tcp*[h]{either \(\Tt\) is not valid or item (b) Def~\ref{def:observation-table-closed-consistent} does not hold here if \(\Tt\) is not consistent by line 4 }){
				\(K \gets K + 1\)\;
			}
			\continue{}\tcp*[l]{is equivalent to goto line 2}
		}
		Build the conjecture \(\Aa_{\Tt} \) from \(\Tt\)\;
		\lIf{Teacher confirms \(L(\Aa_{\Tt}) = L\)}{ \Return \(\Aa_{\Tt}\)}
		\Else{
			Teacher returns counterexample \(w\)\;
			\(S \gets S \cup \{\text{prefixes of } N^{\rl}(w)\}\)\;
		}
	}
\end{algorithm}
\section{Normal Teacher Algorithm -- Additional Material}

\begin{theoremstar}[\ref{theorem:normal_teacher}]
	The Normal Teacher algorithm terminates and returns a strict acceptor for \(L\) of constant \(K\geq K_{L}\) and reset function \(R:\mathbb{R}_{\geq 0}\to [0,K]\setminus\mathbb{N}_{>0}\).
 	Moreover, the algorithm processes \(O(2^{H^{3}m2^m|\Sigma|}H)\) observation tables and performs \(O(2^{H^{3}m2^m|\Sigma|}H^3m2^m|\Sigma|)\) membership queries where \(H=K_{L}+n\), \(n\) is the number of equivalence classes for \(\approx^{L}\) and \(m\) the length of the longest counterexample.

\end{theoremstar}
\begin{proof}
	Termination follows from the termination of the Smart Teacher algorithm, since one of the branches corresponds to an execution of that algorithm.

	We now establish the complexity bounds.
	An execution of the Normal Teacher algorithms gives a tree of OTs.
	Its height \(H\) is at most \(K_{L}+n\) (this bound comes from the Smart Teacher algorithm).
	In a branch of the tree, between two consecutive OTs, the values \(|E|\) and \(K\) are increased by at most one.
	Therefore, \(|E|\leq H\) and \(K\leq H\).
	Between two consecutive OTs in the branch, \(|S|\) can increase by at most \(m2^m\).
	This corresponds to processing a counterexample.
	The counterexample has length at most \(m\) thus, it has at most \(2^{m}\) possible normal forms and each of them has at most \(m\) prefixes.
	Therefore, at worst processing a counterexmple adds \(m2^m\) new words to \(S\).
	Hence, \(|S| \leq Hm2^m\).
	Since \(|\Sigma_{K}|=2(K+1)|\Sigma|\), the number of cells in each OT of the execution is \(|(S\cup S\Sigma_K)E|=|S|(1+2(K+1)|\Sigma|)|E|\).
	We thus find that the number of cells in each OT is bounded by \(H^2m2^m(1+2(H+1)|\Sigma|)=O(H^3m2^m|\Sigma|)\).

	Let \(B\) be a bound on the width of the branching.
	In each branch the number of membership queries is equal to the number of cells of the final OT (the leaf).
	The number of branches (which is equal to the number of leafs of the tree) is bounded by \(B^H\).
	Therefore, the total number of membership queries in an execution is bounded by \(B^HO(H^3m2^m|\Sigma|)\).
	A bound on the total number of processed OTs is \(HB^H\).

	Next we calculate a bound \(B\) for the branching.
	Branching happens after each application of an action {\upshape 1--4} and after processing a counterexample.
	\begin{enumerate}
		\item After an application of action~{\upshape 1} we add \(|\Sigma_{K}|\) new extensions and each one of them is a \(|E|\)-dimensional vector.
		      Therefore, we add \(|\Sigma_{K}||E|\) new cells and for each one them we need to chose whether we reset or not.
		      Thus, the width of the branching is at most \(2^{|\Sigma_{K}||E|}\).
		\item After an application of action~{\upshape 2} we add a new column \((t\cdot a)e\).
		      For every \(s\in S\cup S\Sigma_{K}\) we need to assign a value to \(\Tt(s, (t\cdot a)e)(2)\).
		      Since for every \(s\in S\), we should have \(\Tt(s, (t\cdot a)e)= \Tt(s(t\cdot a), e)\), choice only happens for \(s\in S\Sigma_{K}\).
		      Hence, the width of the branching is at most \(2^{|S||\Sigma_{K}|}\).
		\item After increasing the constant (action~{\upshape 3} or~{\upshape 4}) we add \(2|S|\) new extensions.
		      Thus, the width of the branching is at most \(2^{2|S||E|}\).
		\item After processing a counterexample we add at most \(m2^{m}(1+|\Sigma_{K}|)\) to \( S\cup S\Sigma_{K}\).
		      Hence, the width of the branching is at most \(2^{m2^{m}(1+|\Sigma_{K}|)|E|}\).
	\end{enumerate}
	Hence, the branching degree of the tree is \(O(2^{H^{2}m2^m|\Sigma|})\).
	Finally, we find that the total number of membership queries in an execution is \(O(2^{H^{3}m2^m|\Sigma|}H^3m2^m|\Sigma|)\) and the number of processed OTs is \(O(2^{H^{3}m2^m|\Sigma|}H)\).
	\qed
\end{proof}

In any branch where the reset assignments do not correspond to the reset function of an acceptor for \(L\) termination is not guaranteed.
Example~\ref{example:badbranch} shows that, under an incorrect reset assignment, an incorrect conjecture may persist indefinitely because the Teacher can keep providing non–half-integral counterexamples from the same equivalence class for \(\approx^{L}\).
It can also happen that we never move to a new conjecture because actions~{\upshape 1--4} are applied indefinitely.
For instance, in the ``never-resetting'' branch action~{\upshape 1} is applied infinitely many times: consider the extension \(u(K+\frac{1}{2}\cdot a )\) of a row word \(u\) such that \(r(u)\) is maximal.
Its row differs from all others since the value \(r(u) + K+\frac{1}{2}\) appears in no existing row.
\begin{example}
	\label{example:badbranch}
	We illustrate the Normal Teacher algorithm on a branch with an incorrect reset assignment.
	The unknown language is \(L=\{(x \cdot a)(y \cdot a)\mid 0<x<1, x+y =2\}\).
	We consider the reset assignment \(r\) that maps every entry to zero; that is, for every OT, we take \(r(s,e)=0\) for every \(s\in S\cup S\Sigma_K\) and \(e\in E\).
	The OTs for this example are illustrated in Table~\ref{tab:exampleYYYY}.
	Notice that all the OTs are valid since the reset assignement is zero everywhere.
	The initial OT is \(\Tt_{0}\) in Table~\ref{tab:exampleYYYY} is closed, consistent and valid.
	Its conjecture gives the empty language and a counterexample is the word \((\frac{1}{2}\cdot a)(1+\frac{1}{2}\cdot a)\).
	After processing the counterexample we get \(\Tt'_0\) (Table~\ref{tab:exampleYYYY}).
	Since \(K=0\), \(\Tt'_0\) is inconsistent (b) since \((\frac{1}{2}\cdot a)(\frac{1}{2}\cdot a)\) and \((\frac{1}{2}\cdot a)(1+\frac{1}{2}\cdot a)\) have different rows.
	We thus increase \(K\) to \(1\).
	\(\Tt'_0\) is inconsistent (a): \(\row(\epsilon) = \row((\frac{1}{2}\cdot a))\) but \(\row((1+\frac{1}{2}\cdot a))\neq \row((\frac{1}{2}\cdot a)(1+\frac{1}{2}\cdot a))\).
	To solve this inconsistency we add \((1+\frac{1}{2}\cdot a)\) to \(E\).
	We get \(\Tt''_0\) (Table~\ref{tab:exampleYYYY}) which is closed and consistent.
	Its conjecture \(\Aa_{\Tt''_0}\) (\figurename~\ref{fig:exampleOT_bad}) is wrong: \((\frac{1}{2}\cdot a)(2\cdot a)\in L(\Aa_{\Tt''_0})\) but \((\frac{1}{2}\cdot a)(2\cdot a)\notin L\).
	Adding \((\frac{1}{2}\cdot a)(2\cdot a)\) to \(S\) gives the OT corresponding to \(\Tt'''_0\) minus the third collumn.
	This OT is inconsistent (a): \(\row(\epsilon) = \row((\frac{1}{2}\cdot a)(2\cdot a))\) but \(\row((\frac{1}{2}\cdot a))\neq \row(((\frac{1}{2}\cdot a)(2\cdot a)(\frac{1}{2}\cdot a))\).
	To solve this inconsistency we add \((\frac{1}{2}\cdot a)(1+\frac{1}{2}\cdot a)\) to \(E\) which gives \(\Tt'''_0\) with \(K=1\).
	Since \(K=1\) and \(\row((\frac{1}{2}\cdot a)(1+\frac{1}{2}\cdot a))\neq \row((\frac{1}{2}\cdot a)(2\cdot a))\) the OT \(\Tt'''_0\) is inconsistent (b) and so we increment \(K\).
	Finally, \(\Tt'''_0\) with \(K=2\) is closed and consistent.
	The conjecture \(\Aa_{\Tt'''_0}\) (\figurename~\ref{fig:exampleOT_bad}) is wrong even though \(L\) and \(L(\Aa_{\Tt'''_0})\) have the same half-integral words.
	Say the Teacher returns the counterexample \(w=(0.2 \cdot a)(1.3 \cdot a)\).
	It has two possible normal forms: \(N_1(w)=(\frac{1}{2} \cdot a)(1+\frac{1}{2} \cdot a)\) and \(N_2(w)=(\frac{1}{2} \cdot a)(1 \cdot a)\).
	By adding the missing normal form \(N_2(w)=(\frac{1}{2} \cdot a)(1 \cdot a)\) (\(N_1(w)=(\frac{1}{2} \cdot a)(1+\frac{1}{2} \cdot a)\) is already in the table) the resulting OT is closed and consistent and induces the same conjecture \(\Aa_{\Tt'''_0}\).
\end{example}

\begin{algorithm}[htb]
	\caption{Learning a strict acceptor with a Normal Teacher}
	\label{alg:normal-teacher}
	\KwOut{A strict acceptor \(\Aa_\Tt\) such that \(L(\Aa_\Tt) = L\)}

	\(K \gets 0\)\;
  \(P_0 \gets \{\Tt_0, \dots, \Tt_n\}\)\tcp*[l]{\(P_0\) comprises all the tables \(\Tt_{i}=(0,\{\epsilon\},\{\epsilon\},T_i)\) each corresponding to a distinct reset assignment. 
Since \(S=\{\epsilon\}\), the reset assignments encoded in \(T_i\) apply to the extensions \(S\Sigma_0\).}%

	\For{\(i=0,1,\ldots\)}{
		\(P_{i+1} \gets \emptyset\)\;
		\ForEach{\(\Tt = (K,S,E,T) \in P_i\)}{
			\(P_i \gets P_i \backslash \{\Tt\}\)\;
			\uIf{\(\Tt\) is not closed or not consistent or not valid}{
				\uIf{\(\exists w \in S\Sigma_K,\, \forall s \in S,\, \row(w) \neq \row(s)\)}{
					\(S \gets S \cup \{w\}\)\;
				} \uElseIf{\(\exists s_1, s_2 \in S,\, \exists (t \cdot a) \in \Sigma_K, \, \exists e \in E,\,\row(s_1) = \row(s_2) \land{}\) \(\row(s_1(t \cdot a))(e) \neq \row(s_2(t \cdot a))(e)\)}{
					\(E \gets E \cup \{(t \cdot a)e\}\)\;
				} \Else(\tcp*[h]{either \(\Tt\) is not valid or item (b) Def~\ref{def:observation-table-closed-consistent} does not hold here since \(\Tt\) is not consistent by line 4.}){
					\(K \gets K + 1\)\;
				}
			} \Else{
				Build the conjecture \(\Aa_{\Tt} \) from \(\Tt\)\;
				\lIf{Teacher confirms \(L(\Aa_{\Tt}) = L\)}{\Return \(\Aa_{\Tt}\)}
				\Else{
					Teacher returns counterexample \(w\)\;
					Compute the set \(\{N_1(w),\ldots,N_k(w)\}\) of all normal forms of \(w\) for some \(k>0\)\;
					\(S \gets S \cup \bigcup_{j=1}^k\{\text{prefixes of } N_j(w)\}\)\;
				}
			}
			\(\Tt \gets (K,S,E,T)\)\tcp*[l]{build \(\Tt\) with \(S\), \(E\) and \(K\)}
			Add to \(P_{i+1}\) a copy of \(\Tt\) for every extending reset assignment\;
		}
	}
\end{algorithm}

\begin{table}
	\caption{OTs for Example~\ref{example:badbranch}.
		Each cell displays only the first scalar \(\Tt(s,e)(1)\).
		The second value \(\Tt(s,e)(2)\) which encodes the reset assignment, is always zero and has been omitted for readability.
		We use \(\star\) to mean all the extensions not appearing in \(S\).
		We use bold when the values are the same for multiple rows.
		We have \(e=(1+\frac{1}{2}\cdot a)\) and \(e'=(\frac{1}{2}\cdot a)(1+\frac{1}{2}\cdot a)\).
	}
	\label{tab:exampleYYYY}
	\centering
	\begin{minipage}{0.12\textwidth}
		\centering
		\subcaption{\(K=0\)}
		\begin{tabular}{|c|@{\hspace{2pt}}c@{\hspace{2pt}}|}
			\hline
			\(\Tt_{0}\)              & \(\epsilon\) \\
			\hline
			\(\epsilon\)             & \(0\)        \\
			\hline
			\((0\cdot a)\)           & \(0\)
			\\
			\((\frac{1}{2}\cdot a)\) & \(0\)
			\\
			\hline
		\end{tabular}
	\end{minipage}%
	\begin{minipage}{0.23\textwidth}
		\centering
		\subcaption{\(K=0\) or \(K=1\)}
		\begin{tabular}{|c|@{\hspace{2pt}}c@{\hspace{2pt}}|}
			\hline
			\(\Tt'_0\)                                      & \(\epsilon\)   \\
			\hline
			\(\epsilon\)                                    & \(0\)          \\
			\((\frac{1}{2}\cdot a)\)                        & \(0\)
			\\
			\((\frac{1}{2}\cdot a)(1+ \frac{1}{2}\cdot a)\) & \(1\)          \\
			\hline
			\(\star\)                                       & \(\mathbf{0}\) \\
			\hline
		\end{tabular}
	\end{minipage}%
	\begin{minipage}{0.26\textwidth}
		\centering
		\subcaption{\(K=1\)}
		\begin{tabular}{|c|@{\hspace{2pt}}c@{\hspace{2pt}}|@{\hspace{2pt}}c@{\hspace{2pt}}|}
			\hline
			\(\Tt''_0\)                                     & \(\epsilon\)   & \(e\)          \\
			\hline
			\(\epsilon\)                                    & \(0\)          & \(0\)          \\
			\((\frac{1}{2}\cdot a)\)                        & \(0\)          & \(1\)
			\\
			\((\frac{1}{2}\cdot a)(1+ \frac{1}{2}\cdot a)\) & \(1\)          & \(0\)          \\
			\hline
			\(\star\)                                       & \(\mathbf{0}\) & \(\mathbf{0}\) \\
			\hline
		\end{tabular}
	\end{minipage}%
	\begin{minipage}{0.28\textwidth}
		\centering
		\subcaption{\(K=1\) or \(K=2\)}
		\begin{tabular}{|c|@{\hspace{2pt}}c@{\hspace{2pt}}|@{\hspace{2pt}}c@{\hspace{2pt}}|@{\hspace{2pt}}c@{\hspace{2pt}}|}
			\hline
			\(\Tt'''_0\)                                    & \(\epsilon\)   & \(e\)          & \(e'\)         \\
			\hline
			\(\epsilon\)                                    & \(0\)          & \(0\)          & \(1\)          \\
			\((\frac{1}{2}\cdot a)\)                        & \(0\)          & \(1\)          & \(0\)
			\\
			\((\frac{1}{2}\cdot a)(1+ \frac{1}{2}\cdot a)\) & \(1\)          & \(0\)          & \(0\)          \\
			\((\frac{1}{2}\cdot a)(2\cdot a)\)              & \(0\)          & \(0\)          & \(0\)          \\

			\hline
			\(\star\)                                       & \(\mathbf{0}\) & \(\mathbf{0}\) & \(\mathbf{0}\) \\
			\hline
		\end{tabular}
	\end{minipage}
\end{table}

\begin{figure}[ht]
	\centering
	\begin{tikzpicture}[initial text=\(\Aa_{\Tt_{1}}\),
			every state/.style={
					rectangle,
					rounded corners,
				}]
		\node[state, initial] (q0) at (0,0) {\(\epsilon\)};
		\node[state] (q1) at (4,0) {\((\frac{1}{2}\cdot a)\)};
		\node[state,accepting] (q2) at (8,0) {\((\frac{1}{2}\cdot a)(1+\frac{1}{2}\cdot a)\)};

		\begin{scope}[->, >=stealth, dashed, auto]
			\draw (q0) edge [loop below] node {\(a, {x\notin (0,1) }, 0\)} (q0);
			\draw (q1) edge [bend right] node [above] {\(a, x\neq 2,0\)} (q0);
			\draw (q1) edge node [above, sloped] {\(a, x= 2,0\)} (q2);
			\draw (q2) edge [bend left] node [above] {\(a, 0\leq x,0\)} (q0);
		\end{scope}
		\begin{scope}[->, >=stealth, auto]
			\draw (q0) edge node [below] {\(a,0<x<1 ,1\)} (q1);
		\end{scope}
	\end{tikzpicture}
	\\[10pt]
	\begin{tikzpicture}[initial text=\(\Aa_{\Tt_{2}}\),
			every state/.style={
					rectangle,
					rounded corners,
				}]
		\node[state, initial] (q0) at (0,0) {\(\epsilon\)};
		\node[state] (q1) at (4,0) {\((\frac{1}{2}\cdot a)\)};
		\node[state,accepting] (q2) at (8,0) {\((\frac{1}{2}\cdot a)(1+\frac{1}{2}\cdot a)\)};
		\node[state] (q) at (4,-2) {\((0\cdot a)\)};

		\begin{scope}[->, >=stealth, dashed, auto]
			\draw (q0) edge node [left] {\(a, {x\notin (0,1) }, 0\)} (q);
			\draw (q1) edge node [right, near start] {\(a, x\neq 2,0\)} (q);
			\draw (q1) edge node [above] {\(a, x= 2,0\)} (q2);
			\draw (q2) edge node [right, near end] {\(a, 0\leq x,0\)} (q);
			\draw (q) edge [loop left] node [left] {\(a, {0\leq x}, 0\)} (q);
		\end{scope}
		\begin{scope}[->, >=stealth, auto]
			\draw (q0) edge node [above] {\(a,0<x<1 ,1\)} (q1);
		\end{scope}
	\end{tikzpicture}
	\caption{Conjectures for Example~\ref{example:sm}.}%

	\label{fig:exampleST}
\end{figure}
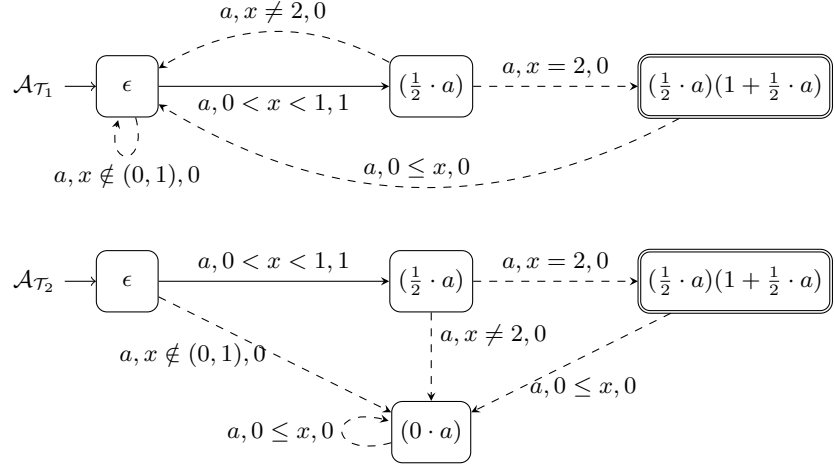

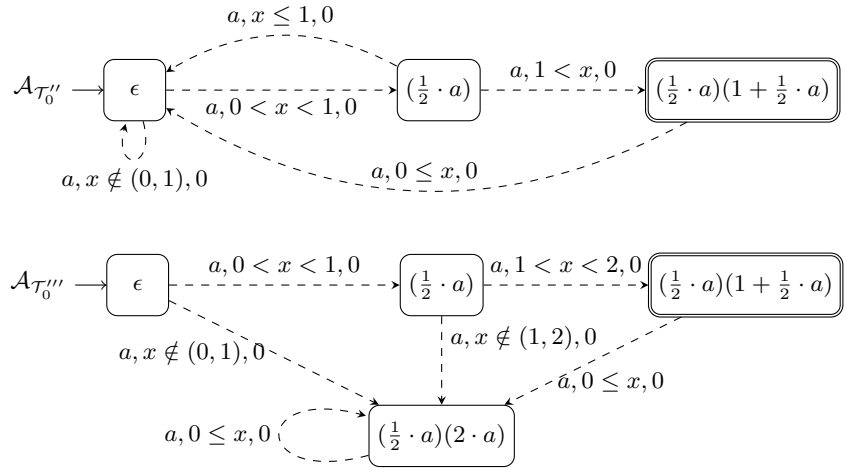
\begin{figure}[ht]
	\centering
	\begin{tikzpicture}[initial text={\(\Aa_{\Tt''_0}\)},
			every state/.style={
					rectangle,
					rounded corners,
					draw=black,
				}]
		\node[state, initial] (q0) at (0,0) {\(\epsilon\)};
		\node[state] (q1) at (4,0) {\((\frac{1}{2}\cdot a)\)};
		\node[state,accepting] (q2) at (8,0) {\((\frac{1}{2}\cdot a)(1+\frac{1}{2}\cdot a)\)};

		\begin{scope}[->, >=stealth, dashed, auto]
			\draw (q0) edge [loop below] node {\(a, {x\notin (0,1) }, 0\)} (q0);
			\draw (q0) edge node [below] {\(a,0<x<1 ,0\)} (q1);
			\draw (q1) edge [bend right] node [above] {\(a, x\leq 1 ,0\)} (q0);
			\draw (q1) edge node [above, sloped] {\(a, 1<x,0\)} (q2);
			\draw (q2) edge [bend left] node [above] {\(a, 0 \leq x,0\)} (q0);
		\end{scope}
	\end{tikzpicture}
	\\[10pt]
	\begin{tikzpicture}[initial text={\(\Aa_{\Tt'''_0}\)},
			every state/.style={
					rectangle,
					rounded corners,
					draw=black,
				}
		]
		\node[state, initial] (q0) at (0,0) {\(\epsilon\)};
		\node[state] (q1) at (4,0) {\((\frac{1}{2}\cdot a)\)};
		\node[state,accepting] (q2) at (8,0) {\((\frac{1}{2}\cdot a)(1+\frac{1}{2}\cdot a)\)};
		\node[state] (q) at (4,-2) {\((\frac{1}{2}\cdot a)(2\cdot a)\)};

		\begin{scope}[->, >=stealth, dashed, auto]
			\draw (q0) edge node [left] {\(a, {x\notin (0,1) }, 0\)} (q);
			\draw (q0) edge node [above] {\(a,0<x<1 ,0\)} (q1);
			\draw (q1) edge node [right, near start] {\(a, x\notin (1,2),0\)} (q);
			\draw (q1) edge node [above] {\(a, 1<x <2,0\)} (q2);
			\draw (q2) edge node [right, near end] {\(a, 0 \leq x,0\)} (q);
			\draw (q) edge [loop left] node [left] {\(a, {0 \leq x}, 0\)} (q);
		\end{scope}
	\end{tikzpicture}
	\caption{Conjectures for Example~\ref{example:badbranch}.}
	\label{fig:exampleOT_bad}

\end{figure}

\end{document}